\documentclass[10pt,a4paper]{article}
\usepackage{amssymb,amsmath}
\usepackage{amsfonts}
\usepackage{amsthm}
\usepackage{graphicx}
\usepackage{tabularx,ragged2e,booktabs,caption}
\usepackage{multirow}
\usepackage{subcaption}
\usepackage{epstopdf}
\numberwithin{equation}{section}

\usepackage{graphicx}
  
 \newtheorem{thm}{Theorem}[section]
 \newtheorem{prop}{Proposition}[section]
 \theoremstyle{plain}
 \newtheorem{lem}[thm]{Lemma} 
 
 \theoremstyle{definition}
 \newtheorem{defn}[thm]{Definition}
 
 \theoremstyle{remark}
 \newtheorem{rem}[thm]{Remark}
 
\DeclareMathOperator{\Tr}{Tr}
\usepackage{amssymb}

\usepackage{amssymb}
\usepackage[left=2cm,right=2cm,top=2cm,bottom=2cm]{geometry}

\newtheorem{definition}{Definition}[section]

\newcommand{\sumai}{\sum_{i=1}^N}
\newcommand{\sumaj}{\sum_{j=1}^N}
\newcommand{\sumak}{\sum_{k=1}^N}

\newcommand{\sumajnoi}{\sum_{j \neq i}^N}
\newcommand{\sumaknoi}{\sum_{k \neq i}^N}

\newcommand\keywords[1]{%
    \begingroup
    \let\and\\
    \par
    \noindent\emph{Keywords:} #1\par
    \endgroup
}

\newcommand\msc[1]{%
    \begingroup
    \let\and\\
    \par
    \noindent\emph{2010 MSC:} #1\par
    \endgroup
}

\makeatother

\author{P. Amster$^{1,2}$ and A.P. Mogni$^1$}
\title{On a pricing problem for a multi-asset option \\ with general transaction costs}
\begin{document}

\date{}
\maketitle
\begin{center}
$^1$ Departamento de Matem\'atica, \\
Facultad de Ciencias Exactas y Naturales\\
Universidad de Buenos Aires
and \\
$^2$ IMAS - CONICET\\
Ciudad Universitaria, Pabell\'on I,
1428 Buenos Aires, Argentina \\
{\sl E-mails}: pamster@dm.uba.ar --- amogni@dm.uba.ar
\end{center}

%
%


\begin{abstract}

\noindent We consider a Black-Scholes type equation arising on a pricing model for a multi-asset option with general transaction costs. 
The pioneering work of Leland is thus extended in two different ways: on the one hand, the problem is multi-dimensional since it involves different underlying assets; on the other hand, the transaction costs are not assumed to be constant (i.e. a fixed proportion of the traded quantity). In this work, we generalize Leland's condition and prove the existence of a viscosity solution for the corresponding fully nonlinear initial value problem using Perron method. Moreover, we develop a numerical ADI scheme to find an approximated solution. We apply this method on a specific multi-asset derivative and we obtain the option price under different pricing scenarios.\\


\msc{35K20, 35K55, 91G20, 91G60}

\end{abstract}

\maketitle
\section{Introduction}

\noindent The Black-Scholes model \cite{black1973pricing} relies on different assumptions such as constant values of volatility and interest rates, the non-existence of dividend yields, the efficiency of the markets and the non-existence of transaction costs, among others. Following Leland's approach \cite{leland1985option}, transaction costs can be included in the pricing methodology by applying a discrete-time replicating strategy. A nonlinear partial differential equation is obtained for the option price, which is denoted by $V\left(S,t\right)$; namely,

\begin{align}
\frac{\partial V}{\partial t}+\frac{1}{2} \hat{\sigma} \left( S \frac{\partial^2 V}{\partial S^2} \right)^2 S^2  \frac{\partial^2 V}{\partial S^2} +rS\frac{\partial V}{\partial S}-rV=0,
\end{align}

\noindent where $\hat{\sigma}$ is defined based upon the transaction costs function. 
For example, if transaction costs are defined by a constant rate $C_0$, then $\hat{\sigma}$ is given by

\[
\hat{\sigma} 
\left( S \frac{\partial^2 V}{\partial S^2} \right)^2=
\sigma^2 \left(1-\hbox{Le} \, \hbox{sgn}\left(S \frac{\partial^2 V}{\partial S^2} \right) \right) = \left\{
													\begin{array}{ccc}
													\sigma^2 \left(1-\hbox{Le}\right) & \hbox{if} & \frac{\partial^2 													V}{\partial S^2}>0\\
													\sigma^2 \left(1+\hbox{Le}\right) & \hbox{if} & \frac{\partial^2 													V}{\partial S^2}<0\\
													\end{array}
													\right.
\]

\noindent where $\hbox{Le}=\sqrt{\frac{2}{\pi}}\frac{C_0}{\sigma\sqrt{\Delta t}}$ is the Leland number.\\
 
\noindent The original approach was extended by different authors. A discrete approximation is studied in \cite{boyle1992option} by developing a binomial option pricing model with constant transaction costs. The generalization of Leland's methodology for a portfolio of options is presented in \cite{hoggard1994hedging} and the existence of solution is studied in \cite{imai2006hoggard}. In \cite{grossinho2009note}, a method of upper and lower solutions is used to study the original stationary problem. Also, an analysis of the original hedging strategy is found in \cite{grandits2001leland} and a modification of the strategy is considered in \cite{lepinette2012modified} to guarantee that the approximation error vanishes in the limit.\\

\noindent Different choices of transaction costs functions lead to variations on the nonlinear term of the partial differential equation. In \cite{amster2005black}, the authors propose a non-increasing linear function and find solutions for the stationary problem. In \cite{vsevvcovivc2016analysis}, the concept of transaction costs function is generalized and the so-called \emph{mean value modification of the transaction costs function} is developed. This transformation allows the authors to formulate a general one-dimensional Black-Scholes equation by solving the equivalent quasilinear Gamma equation. Moreover, viscosity solutions have been studied in the nonlinear problems that arises from including transaction costs in the option pricing framework . The seminal work of \cite{davis1993european} finds the option price by comparing the maximum utilities available to the writer   leading to solve two stochastic optimal control problems. Unique viscosity solutions are found as the value functions of these problems. Moreover, the work of \cite{barles1998option} uses a utility function with an asymptotic analysis of partial differential equations to quantify the dependence on preferences in European call option problem.\\

\noindent The main distinctive aspect in the above-cited works is that they all consider only one asset within the partial differential equation. In \cite{zakamouline2008hedging} and \cite{zakamulin2008option}, the author generalizes the Leland approach in order to cover different types of multi-asset options, developing the nonlinear partial differential equation and solving numerically a list of examples.\\

\noindent In this work, we prove the existence of a viscosity solution for the problem of pricing a multi-asset option with a general transaction costs function. We derive the following nonlinear problem

\begin{align}
-V_{\tau}+\mathcal{L}V = G\left(V\right) \quad  &\hbox{in} \quad \Omega\times\left[0,T\right] \nonumber \\
V\left(x_1,...,x_N,0 \right)= V_0\left(x_1,...,x_N\right) \quad &\hbox{in} \quad \Omega \label{problema_completo}
\end{align}

\noindent where $\Omega = \mathbb{R}^N$, $V$ is the option price, $\mathcal{L}$ is an elliptic operator, $G$ is a nonlinear term and $V_0$ is the initial condition. This problem can be rewritten in terms of a nonlinear elliptic operator $F$ as

\begin{align}
-V_{\tau}+FV = 0 \quad  &\hbox{in} \quad \Omega\times\left[0,T\right] \nonumber \\
V\left(x_1,...,x_N,0 \right)= V_0\left(x_1,...,x_N\right) \quad &\hbox{in} \quad \Omega \label{problema_with_F}
\end{align}

\noindent This presentation helps us to introduce the Perron method to find a viscosity solution. Indeed, in our work we show that a generalization of Leland's condition is required such that the nonlinear operator $F$ becomes degenerate elliptic and a solution can be found. By defining properly the sub and supersolutions of problem \eqref{problema_with_F} and recalling a comparison principle, we use Perron method to derive the existence of solution.\\

\noindent In the second part of the work, we develop a numerical approach in order to find a solution using an iterative method. For this purpose, the Alternating Difference Implicit (ADI) scheme is selected within the family of splitting operators. Different works \cite{in2010adi,in2007stability,mckee1996alternating,mckee1970alternating} study the applicability of this approach to deal with the mixed derivatives terms of the discretization. On multidimensional problems, the ADI method allows to solve efficiently the PDE problem by applying a tridiagonal matrix algorithm in comparison to the classical Crank-Nicholson scheme. In this section we provide results regarding the convergence of the numerical scheme, the sensitivity of the final output to the choice of timing parameters and the impact of the transaction costs in the option price.\\

\noindent The structure of the paper is as follows. In Section 2 we derive the nonlinear PDE that explains the dynamics of the option price for a multi-asset derivative considering a general transaction costs function. In Section 3 we apply all the necessary steps to prove the existence of a viscosity solution using Perron method. Finally, in Section 4 we develop the ADI framework in order to find a strong solution and price a specific multi-asset derivative. 

\section{PDE derivation for multiple assets and general transaction costs function}

\noindent Let $\Pi$ be the portfolio that contains $\delta_i$ of asset $S_i$ and an option $V$ over those assets at time $t$. This portfolio can be represented by 
\begin{equation}
\Pi=V+\sumai \delta_i S_i. \label{Pi}
\end{equation}

\noindent If we define $\Delta$ as the one-step variation of a process (i.e $\Delta y_t=y_t - y_{t-1}$), by applying the It\^{o}'s formula over $V$, we get

\begin{equation}
\Delta V = \frac{\partial V}{\partial t} \Delta t + \sumai \frac{\partial V}{\partial S_i} \Delta S_i + \frac{1}{2} \sumai \sumaj \sigma_i \sigma_j \rho_{ij} S_i S_j \frac{\partial^2 V}{\partial S_i \partial S_j} \Delta t. \label{ito}
\end{equation}

\noindent Transaction costs appear when calculating $\Delta \Pi$, which expresses the variation of the portfolio at each time $t$. Specifically, the variation of the portfolio is represented by 
\begin{equation}
\Delta \Pi=\Delta \left(V+\sumai \delta_i S_i \right)+\sumai \Delta TC_i,
\end{equation}

\noindent where $\Delta TC_i$ is the amount of transaction costs when buying or selling $\delta_i$ assets of $S_i$. By taking $\delta_i=-V_{S_i}$, we obtain  

\begin{equation}
\Delta \Pi = \Delta V - \sumai \frac{\partial V}{\partial S_i} \Delta S_i - \sumai \Delta TC_{i}.\label{dpi}
\end{equation}

\noindent Following the approach in \cite{vsevvcovivc2016analysis}, 
it is seen that 

\begin{equation}
\Delta TC_{i}= S_i \, C\left(\left| \Delta \delta_i \right| \right) \left| \Delta \delta_i \right|,\label{dv}
\end{equation}

\noindent where $C$ is the transaction costs function. By defining $r_{TC}^i$ to be the expected value of the change of the transaction costs per unit time interval $\Delta t$ and price $S_i$, we see that

$$
r_{TC}^i = \frac{E \left[ \Delta TC_i \right]}{S_i \Delta t} = \frac{E\left[ C\left(| \Delta \delta_i | \right) \left| \Delta \delta_i \right| \right]}{\Delta t}.
$$
\noindent Thus, we approximate the transaction costs by the expected value of the transaction costs function applied to the amount of assets bought or sold and multiplied by these amount again. This value is then multiplied by the price of asset $S_i$ in order to get a transaction cost in money terms.\\

\noindent Applying (\ref{dv}) in (\ref{dpi}) and using $\Delta V$, we obtain

\begin{equation}
\Delta \Pi = \left( \frac{\partial V}{\partial t} + \frac{1}{2} \sumai \sumaj \sigma_i \sigma_j \rho_{ij} S_i S_j \frac{\partial^2 V}{\partial S_i \partial S_j} \right) \Delta t - \sumai S_i \, r_{TC}^i \, \Delta t. \label{dpimas}
\end{equation}

\noindent 
From the assumption $\Delta \Pi = r\Pi \Delta t$ and 
(\ref{dpimas}), we obtain

\begin{equation}
rV + \sumai r^i_{TC} \, S_i =  \frac{\partial V}{\partial t} + \frac{1}{2} \sumai \sumaj \sigma_i \sigma_j \rho_{ij} S_i S_j \frac{\partial^2 V}{\partial S_i \partial S_j} + r \sumai \frac{\partial V}{\partial S_i} S_i \label{eq1}
\end{equation}

\noindent where $r_{TC}^i S_i = \frac{E \left[ \Delta TC_i \right]}{\Delta t} = \frac{E\left[ C\left( | \Delta \delta_i | \right) | \Delta \delta_i | S_i \right]}{\Delta t}$.\\

\noindent Equation (\ref{eq1}) is the nonlinear PDE that represents the behaviour of the option price for a multi-asset option when defining  a general transaction costs function. In order to get the complete expression of the PDE, we have to calculate $\delta_i$. From previous steps we know that

$$
\Delta \delta_i = - \Delta \frac{\partial V}{\partial S_i} \sim \sumaj \frac{\partial^2 V}{\partial S_i \partial S_j} \Delta S_j
$$
\noindent taking only the terms with order ${\Delta t}^{1/2}$. Noting that 

$$
\Delta S_j \sim \sigma_j S_j \phi_j \sqrt{\Delta t},
$$
\noindent with $\phi_j$ being a standard normal variable, we find that

$$
\left| \Delta \delta_i \right| = \left| \sumaj \frac{\partial^2 V}{\partial S_i \partial S_j} \Delta S_j \right| =  \left| \sumaj \frac{\partial^2 V}{\partial S_i \partial S_j} \sqrt{\Delta t} \, \sigma_j \, S_j \, \phi_j \, \right| = \sqrt{\Delta t} \, \left| \sumaj \frac{\partial^2 V}{\partial S_i \partial S_j} \, \sigma_j \, S_j \, \phi_j \, \right|.
$$
\noindent Setting $\Phi_i= \sumaj \frac{\partial^2 V}{\partial S_i \partial S_j} \, \sigma_j \, S_j \, \phi_j $, we obtain
that $\Phi_i \sim N\left(0,\Theta_i \right)$ with

\begin{align}
\Theta_i =\sumaj \sumak \frac{\partial^2 V}{\partial S_i \partial S_j} \frac{\partial^2 V}{\partial S_i \partial S_k} \sigma_j \sigma_k \rho_{jk} S_j S_k .   \label{theta1}
\end{align}

\noindent where $\rho_{jk}$ is the correlation parameter between $\phi_j$ and $\phi_k$. Therefore,

\begin{equation}
r_{TC}^i S_i = \frac{E \left[ \Delta TC \right]}{\Delta t} = \frac{E\left[ C\left( | \Delta \delta_i | \right) | \Delta \delta_i | S_i \right]}{\Delta t} = \frac{\sqrt{\Delta t} \, E\left[ C\left( \sqrt{\Delta t} \, \left| \Phi_i \right|  \right)  \left| \Phi_i \right| \ S_i \right]}{\Delta t} = \frac{S_i}{\sqrt{\Delta t}}E\left[ C\left( \sqrt{\Delta t}  \, \left| \Phi_i \right| \right) \left| \Phi_i \right|  \right]. \label{rtc}
\end{equation}

\noindent Using \eqref{rtc} in \eqref{eq1}, we find the following nonlinear PDE which models the dynamic of a multi-asset option.

\begin{equation}
rV + \sumai \frac{S_i}{\sqrt{\Delta t}}E\left[ C\left( \sqrt{\Delta t}  \, \left| \Phi_i \right| \right) \left| \Phi_i \right|  \right] = \frac{\partial V}{\partial t} + \frac{1}{2} \sumai \sumaj \sigma_i \sigma_j \rho_{ij} S_i S_j \frac{\partial^2 V}{\partial S_i \partial S_j} + r \sumai \frac{\partial V}{\partial S_i} S_i.\\  \label{eq2}
\end{equation}

\section{Existence of solution for the resulting PDE}\label{demostracion}
\subsection{Defining the nonlinear problem}

\noindent Let $C$ be a measurable bounded transaction costs function such that $C:\mathbb{R}_0^{+} \rightarrow \mathbb{R}_0^{+}$, $C \in L^{2}\left(\mathbb{R}_0^{+} \right)$ and let $\overline{C}, \underline{C}>0$ be such that $\underline{C}<C\left(x \right)<\overline{C}$ for every $x \in \mathbb{R}_0^{+}$. Moreover, we denote $\Omega = \mathbb{R}^N, \Omega^+ = \mathbb{R}^N_+, \Omega_T = \left[0,T\right] \times \mathbb{R}^N$ and $\Omega_T^+ = \left[0,T\right] \times \mathbb{R}^N_+$. Let us define $G$ to be the nonlinear operator 

\begin{align}
G\left(S,D^2 V\right) &= \sum_{i=1}^{N} \frac{S_i}{\sqrt{\Delta t}} E\left[C\left(\sqrt{\Delta t} \left| \Phi_i \right| \right) \left| \Phi_i \right| \right]\\
&= \sum_{i=1}^{N} \frac{S_i}{\sqrt{\Delta t}} \sqrt{\frac{2}{\pi}} \, 2 \, \sqrt{\Theta_i} \int_{0}^{+\infty} C\left(\sqrt{\Delta t \, 2 \, \Theta_i } y \right) \,  y \, e^{-y^2}  \, dy
\end{align}

\noindent where $\Theta_i$ is given by

\begin{align}
\Theta_i =\sumaj \sumak \frac{\partial^2 V}{\partial S_i \partial S_j} \frac{\partial^2 V}{\partial S_i \partial S_k} \sigma_j \sigma_k \rho_{jk} S_j S_k .   \label{theta1_sec2}
\end{align}

\noindent where $\rho_{jk}$ is the correlation parameter between $\phi_j$ and $\phi_k$, both standard normal variables. Moreover, let us denote $L$ to be the following parabolic operator

\begin{equation}
L\left(\tau,S,V\right)= - rV -  \frac{\partial V}{\partial \tau} + \frac{1}{2} \sumai \sumaj \sigma_i \sigma_j \rho_{ij} S_i S_j \frac{\partial^2 V}{\partial S_i \partial S_j} + r \sumai \frac{\partial V}{\partial S_i} S_i.\ \label{eq2_sec2}
\end{equation}

\noindent Then, we define the nonlinear PDE for the problem of pricing a multi-asset option with general transaction costs as of

\begin{alignat}{2}
\label{problema_original}
\mathcal{L}\left(\tau,S_1,...,S_N,V\right) &= G\left(S_1,...,S_N,D^2 V\right) \quad  &&\hbox{in} \quad \Omega^+\times\left[0,T\right] \nonumber\\
V\left(0,S_1,...,S_N \right) &= V_0\left(S_1,...,S_N\right) \quad  &&\hbox{in} \quad \Omega^+
\end{alignat}

\noindent Our objective is to find a viscosity solution of problem \eqref{problema_original}. For this purpose, we will rewrite problem \eqref{problema_original} to match with the notation of \cite{imbert2013introduction}. Hence, we redefine our nonlinear parabolic equation as

\begin{align}
\frac{\partial V}{\partial \tau} + F\left(\tau,S,V,DV,D^2V \right) = 0 \label{nonlinear-equation}
\end{align}

\noindent where 

\begin{align}
F\left(\tau,S,V,DV,D^2V \right) =  -\frac{1}{2} \sumai \sumaj \sigma_i \sigma_j \rho_{ij} S_i S_j \frac{\partial^2 V}{\partial S_i \partial S_j} - r \sumai \frac{\partial V}{\partial S_i} S_i + rV + G\left(S,D^2V\right).\label{eq-F}
\end{align}

\begin{rem}
Equation \eqref{nonlinear-equation} can be rewritten following a matricial form. If we denote the matrix $A$ as

\begin{align}
\left(A\right)_{ij} = \sigma_i \sigma_j \rho_{ij} S_i S_j
\end{align}

\noindent then the function $F$ can be set as

\begin{align}
F\left(\tau,S,V,DV,D^2V \right) =  -\frac{1}{2} \, \, tr\left(A \, D^2V \right) - r DV \cdot S + rV + G\left(S,D^2V\right) \label{new_F}
\end{align}

\noindent For the nonlinear term that correspond to the function $G$ we first note that the value of $\Theta_i$ is equivalent to the i-th term of the diagonal of the product $D^2V \, A \, D^2V $, i.e.

\begin{align}
\Theta_i = \left(D^2V \, A \, D^2V \right)_{ii} \label{new_theta}
\end{align}

\noindent Then, the function $G$ noted in a matricial form as of

\begin{align}
G\left(S,D^2 V\right) = \sum_{i=1}^{N} \frac{S_i}{\sqrt{\Delta t}} \sqrt{\frac{2}{\pi}} \, 2 \, \sqrt{\left(D^2V \, A \, D^2V \right)_{ii}} \int_{0}^{+\infty} C\left(\sqrt{\Delta t \, 2 \, \left(D^2V \, A \, D^2V \right)_{ii} } \, y \right) \,  y \, e^{-y^2}  \, dy \label{new_G}
\end{align}

\end{rem}

\subsection{Degenerate Ellipticity and Leland's condition}

\subsubsection{Deriving the conditions}

\noindent We are going to prove the existence of a viscosity solution of problem \eqref{nonlinear-equation} using Perron method. The main idea of the method is to construct a subsolution $u^-$ and a supersolution $u^+$ of the nonlinear parabolic equation such that $u^- \leq u^+$. Moreover, it is possible to construct a subsolution $u$ lying between $u^-$ and $u^+$ and see that the lower semi-continuous envelope of the subsolution $u$ is a supersolution. Before applying Perron method, we need to set different conditions on the nonlinear operator $F$. Let us start by presenting the definition of degenerate ellipticity. For this purpose, we will denote $\mathbb{S}_N$ as the space of N-dimensional square symmetric matrices.

\begin{definition}\label{def-deg-el}
\noindent A nonlinear function $F: \left[0,T\right] \times \Omega^+ \times \mathbb{R} \times \mathbb{R}^N \times \mathbb{S}_N \rightarrow \mathbb{R}$ is degenerate elliptic if 

\begin{align}
X \leq Y \implies F\left(t,x,p,s,X\right) \geq F\left(t,x,p,s,Y\right). \label{deg-el}
\end{align}

\end{definition}

\noindent Given the definition of degenerate ellipticity we have to set the correspondent conditions such that the nonlinear function $F$ follows Condition \eqref{deg-el}. Let us start by denoting the differential of function $F$ with respect to the second derivative component $Y$ as

\begin{align}
D_YF\left(t,x,p,s,B\right) = \frac{\partial F\left(t,x,p,s,Y\right)}{\partial Y}\biggr\rvert_{Y = B}
\end{align}

\noindent By Definition \ref{def-deg-el}, given a positive definite matrix $U$, we want to see that 

\begin{align*}
D_YF\left(t,x,p,s,Y\right)\left(U\right) \leq 0
\end{align*}

\noindent If this condition is fulfilled, we can use the mean value theorem to prove that operator $F$ is degenerate elliptic so

\begin{align}
F\left(t,x,p,s,Y\right) - F\left(t,x,p,s,X\right) &= D_YF\left(t,x,p,s,B\right) \cdot \left(Y-X\right) \label{mvt} \\
 &= 0 \nonumber
\end{align}

\noindent where $B \in \left(X,Y \right)$  and $Y-X$ is a positive definite matrix.\\

\noindent Let us recall the Leland condition which is present in the unidimensional problem with a constant transaction costs function. The aim of the this condition is in fact to define a degenerate elliptic operator such that the matrix of coefficients that correspond to the second derivatives is definite positive. In our work, the generalized Leland condition will act as the same and will be deduced from the following two Lemmas.

\noindent The first Lemma shows that, if the differential matrix $D_YF$ is symmetric, evaluating the differential on any definite positive matrix is equivalent to calculating the trace of the product between the differential matrix and the correspondent definite positive matrix.

\begin{lem}\label{sec-lemma-bicon}
\noindent Let $U$ be a positive definite matrix and $D$ the differential matrix with respect to component Y. Then,  $Tr\left(D \, U\right) = D \left(U \right)$
\end{lem}

\begin{proof}
\noindent Let us see that the result follows by using the definition of the Frobenius inner product. From the definition of the the trace of the product between $D$ and the positive definite matrix $U$ and the symmetry of matrix $D$ we have that

\begin{align*}
\Tr\left(D \, U \right) &= \sumai \sumaj D_{ij} U_{ji}\\
&= \sumaj \sumai D_{ji} U_{ji}
\end{align*}

\noindent Now, we can arrange terms so that

\begin{align*}
\Tr\left(D \, U \right)  = D \cdot U = D \left(U\right)
\end{align*}

\end{proof}

\noindent The second Lemma states that we can characterize the sign of the eigenvalues of the differential matrix $D_YF$ in terms of the sign of the trace of the product between $D_YF$ and a definite positive matrix $U$.

\begin{lem}\label{lemma-bicon}
\noindent Let $U$ be a positive definite matrix. Then $D_YF$ is negative definite if and only if $Tr\left(D_YF \, U\right) \leq 0$ for all $U \geq 0$.
\end{lem}

\begin{proof}
\noindent Let us start observing that as $D_YF$ is a symmetric matrix, there exists a diagonal matrix $\tilde{D}$ and a change of basis matrix $C$ such that $D = C^{-1} \tilde{D} C$. Then, we have that

\begin{align}
\Tr\left(D_YF \, U\right) = \Tr\left(C^{-1} \, \tilde{D} \, C \, U\right) = \Tr\left(C^{-1} \, \tilde{D} \, C \, U \, C^{-1} \, C\right).
\end{align}

\noindent If we denote $W =  C \, U \, C^{-1}$, the previous equation can be rewritten as

\begin{align}
\Tr\left(D_YF \, U\right) =  \Tr\left(C^{-1} \, \tilde{D} \, W \, C \right) = \Tr\left(\tilde{D} \, W \right),
\end{align}

\noindent where $W$ is a positive definite matrix. Using the last equality we can prove our statement. If  $Tr\left(D_YF \, U\right) \leq 0$ for all $U \geq 0$, let us choose a sparse matrix $U$ such that column $j$ corresponds to the standard vector $e_j$. Then, $W = U$ and $\tilde{D} W = \lambda_j$. Using the fact that $Tr\left(D_YF \, U\right) \leq 0$, we deduce that each $\lambda_j < 0$. 

\noindent Let us now suppose that $D_YF$ is negative definite. Then,

\begin{align}
\Tr\left(\tilde{D} \, W \right) = \sumai \tilde{D}_{ii} W_{ii} < 0
\end{align}

\noindent as each $\tilde{D}_{ii}$ are negative and each $W_{ii}$ are positive.

\end{proof}

\noindent Both Lemmas \ref{sec-lemma-bicon} and \ref{lemma-bicon} can be resumed in the following line: If the differential matrix $D_YF$ is symmetric, for all matrix $U \geq 0$ the following equivalences are valid

\begin{align*}
D_YF \leq 0 \iff \Tr\left(D_YF \, U\right) \leq 0 \iff  D_YF \left(U\right) \leq 0
\end{align*}

\noindent Recalling \eqref{mvt}, the matrix $Y-X$ is definite positive so by discarding the dependencies, the inequality becomes

\begin{align}
F_Y - F_X = D_YF \left(Y-X\right).
\end{align}

\noindent Hence, the nonlinear operator $F$ is degenerate elliptic if the differential matrix $D_YF$ is symmetric definite negative. In the following section we will see that the condition of being symmetric definite negative is the generalization of the Leland condition defined for the unidimensional problem with constant transaction costs.

\subsubsection{Differential Matrix calculation}

\noindent In this section we perform the calculations of the differential matrix with respect to the second derivatives of the nonlinear term $F$. Let us recall Equation \eqref{new_F} such that

\begin{align}
F\left(t,x,p,s,B \right) =  -\frac{1}{2} \, \, tr\left(A \, B \right) - r s \cdot S + rp + G\left(S,B\right)
\end{align}

\noindent Then, by applying standard calculations and discarding function dependencies, we have that

\begin{align}
D_YF\left(t,x,p,s,B \right)  = -\frac{\partial}{\partial B}  \, tr\left( \frac{1}{2}  A \, B \right) + \frac{\partial}{\partial B}  G\left(S,B\right) \label{DF}
\end{align}

\noindent The first derivative follows recalling the linearity of the trace function and the symmetry of matrix $A$. Then,


\begin{align}
\frac{\partial}{\partial B}  \, tr\left( \frac{1}{2}  A \, B \right) =  \frac{1}{2} A \label{DF_1}
\end{align}

\noindent The second derivative involves applying the product rule on the transaction costs term. Then,

\begin{align}
\frac{\partial}{\partial B}  G\left(S,B\right) &=\frac{\partial}{\partial B}  \left[\sum_{i=1}^{N} \frac{S_i}{\sqrt{\Delta t}} \sqrt{\frac{2}{\pi}} \, 2 \, \sqrt{\sumaj \sumak B_{ij} \, A_{jk} \, B_{ki}} \int_{0}^{+\infty} C\left(\sqrt{2 \, \Delta t \, \sumaj \sumak B_{ij} \, A_{jk} \, B_{ki} } \, y \right) \,  y \, e^{-y^2}  \, dy \right] \nonumber \\
&= \sum_{i=1}^{N} \frac{S_i}{\sqrt{\Delta t}} \sqrt{\frac{2}{\pi}} \, 2 \left[ \frac{\partial}{\partial B} \sqrt{\sumaj \sumak B_{ij} \, A_{jk} \, B_{ki}} \int_{0}^{+\infty} C\left(\sqrt{2 \, \Delta t \, \sumaj \sumak B_{ij} \, A_{jk} \, B_{ki} } \, y \right) \,  y \, e^{-y^2}  \, dy  \right. \nonumber\\
&+ \left. \sqrt{\sumaj \sumak B_{ij} \, A_{jk} \, B_{ki}} \int_{0}^{+\infty} \frac{\partial}{\partial B}  C\left(\sqrt{2 \, \Delta t \, \sumaj \sumak B_{ij} \, A_{jk} \, B_{ki} } \, y \right) \,  y \, e^{-y^2}  \, dy \right] \label{big-der}
\end{align}

\noindent The above calculation can be solved by analysing two derivatives. The first one correspond to the $\Theta_i$ function defined in \eqref{theta1_sec2}. The calculation of the derivative of this term is done in \ref{appendix} and is given by 

%

\begin{align}
\frac{\partial}{\partial B} \sqrt{\Theta_i} = \frac{1}{2} \Theta_i^{-1/2} \left[ AB + BA \right]. \label{DF_2}
\end{align}

\noindent The second derivative corresponds to the derivative of the transaction costs function $C$ with respect to matrix $B$. Again, the complete calculation is presented in \ref{appendix}. Then, the derivative with respect to matrix $B$ is equal to

%
%
%

\begin{align}
 \frac{\partial}{\partial B}  C\left(\sqrt{2 \, \Delta t \, \left(BAB \right)_{ii} } \, y \right) = C'\left(H_i\left(y\right) \right) \, y \, \left[ AB + BA \right] \, \sqrt{\frac{\Delta t}{2}} \, \Theta_i^{-1/2} \label{DF_3}
\end{align}

\noindent Now, we can write Equation \eqref{DF} as

\begin{align}
D_YF\left(t,x,p,s,B \right) &=  -\frac{1}{2} A + 2 \sumai \frac{S_i}{\sqrt{\Delta t}} \sqrt{\frac{2}{\pi}} \left[ \frac{1}{2} \Theta_i^{-1/2} \left[ AB + BA \right] \int_{0}^{+\infty} C\left(\sqrt{2 \, \Delta t \,\left(BAB \right)_{ii} } \, y \right) \,  y \, e^{-y^2}  \, dy   \right. \nonumber \\
&+ \left. \left[BA + AB \right] \sqrt{\frac{\Delta t}{2}} \int_{0}^{+\infty} C'\left(\sqrt{2 \, \Delta t \, \left(BAB \right)_{ii} } \, y \right) \,  y^2 \, e^{-y^2}  \, dy \right] \nonumber \\
D_YF\left(t,x,p,s,B \right)&=  -\frac{1}{2} A + \left[BA + AB \right] \frac{2}{\sqrt{\Delta t}} \sqrt{\frac{2}{\pi}} \sumai S_i \left[ \frac{1}{2} \Theta_i^{-1/2}  \int_{0}^{+\infty} C\left(\sqrt{2 \, \Delta t \,\left(BAB \right)_{ii} } \, y \right) \,  y \, e^{-y^2}  \, dy  \right. \nonumber \\
&+ \left.  \sqrt{\frac{\Delta t}{2}}\int_{0}^{+\infty} C'\left(\sqrt{2 \, \Delta t \, \left(BAB \right)_{ii} } \, y \right) \,  y^2 \, e^{-y^2}  \, dy \right]. \label{DFB_calc}
\end{align}

\noindent Equation \eqref{DFB_calc} defines the final state of the differential matrix of the nonlinear parabolic operator $F$ with respect to the component of the second derivatives. The generalized Leland's condition found in Lemmas \ref{sec-lemma-bicon} and \ref{lemma-bicon} requires that the differential matrix $D_YF$ is definite negative. In fact, we can check that this condition reduces to the original Leland's condition when fixing $N = 1$ and the function of transaction costs $C$ as constant.

\begin{rem}

\noindent Let us show that effectively our condition reduces to Leland's condition in the unidimensional case with constant transaction costs. For this purpose, we assign $A$, $\Theta$ and $C$ as in the unidimensional case. Then,

\begin{align}
A = S^2 \sigma^2, \quad \Theta = \frac{\partial^2 V}{\partial S^2} \sigma^2 S^2, \quad C\left(\sqrt{2 \, \Delta t \,\left(D^2V A D^2V \right) } \, y \right) = \frac{\tilde{C}}{2}
\end{align}

\noindent If we apply this definitions on Equation \eqref{DFB_calc}, we get that

\begin{align}
D_YF\left(t,x,p,s,\frac{\partial^2 V}{\partial S^2} \right) &= -\frac{1}{2} S^2 \sigma^2 + 2\frac{\partial^2 V}{\partial S^2} S^2\sigma^2 \frac{2 \, S}{\sqrt{\Delta t}} \sqrt{\frac{2}{\pi}} \frac{1}{2} \frac{\tilde{C}}{4} \left(\frac{\partial^2 V}{\partial S^2} S^2\sigma^2\right)^{-1/2} \nonumber \\
&= -\frac{1}{2} S^2 \sigma^2 + S^2\sigma^2 \, \text{sgn}\left( \frac{\partial^2 V}{\partial S^2} \right) \frac{S}{\sqrt{\Delta t}}  \sqrt{\frac{2}{\pi}} \frac{\tilde{C}}{2\sigma S} \nonumber \\
&=  \frac{1}{2} S^2 \sigma^2 \left[ -1 + \frac{\tilde{C}}{\sqrt{\Delta t}} \sqrt{\frac{2}{\pi}} \frac{1}{\sigma} \, \text{sgn}\left( \frac{\partial^2 V}{\partial S^2} \right) \right]
\end{align}

Then, $D_YF$ is negative if and only if

\begin{align}
\frac{\tilde{C}}{\sqrt{\Delta t}} \sqrt{\frac{2}{\pi}} \frac{1}{\sigma} < 1
\end{align}

\end{rem}

\subsection{Perron method for existence of solution}

\noindent Let us start this section by setting the framework to apply the well-known Perron method to derive the existence of a viscosity solution. We will first apply a change of variables so that the nonlinear operator $F$ is defined with constant coefficients. Then, we apply the change of variables

\begin{align*}
x_i = \log\left(S_i\right)
\end{align*}

\noindent so that the nonlinear operator $F$ becomes

\begin{align}
F\left(\tau,x,V,DV,D^2V \right) = - \frac{1}{2} \sumai \sumaj \sigma_i \sigma_j \rho_{ij} \frac{\partial^2 V}{\partial x_i \partial x_j} - \sumai \frac{\partial V}{\partial x_i} \left(r - \frac{\sigma_i^2}{2}\right) + rV + G\left(x,D^2V\right),\label{Fchange}
\end{align}

\noindent and the nonlinear function $G$ becomes

\begin{align}
G\left(x,D^2V\right) = \sum_{i=1}^{N} \frac{e^{x_i}}{\sqrt{\Delta t}} \sqrt{\frac{2}{\pi}} \, 2 \, \sqrt{\Theta_i} \int_{0}^{+\infty} C\left(\sqrt{\Delta t \, 2 \, \Theta_i } y \right) \,  y \, e^{-y^2}  \, dy,\label{G_change}
\end{align}

\noindent with

\begin{align}
\Theta_i= e^{-2x_i} \left[ \sumajnoi \sumaknoi \frac{\partial^2 V}{\partial x_i \partial x_j} \frac{\partial^2 V}{\partial x_i \partial x_k} \sigma_j \sigma_k \rho_{jk} + 2 \, \sumajnoi \frac{\partial^2 V}{\partial x_i \partial x_j} \left(\frac{\partial^2 V}{\partial x_i^2}-\frac{\partial V}{\partial x_i} \right) \sigma_i \sigma_j + \left(\frac{\partial^2 V}{\partial x_i^2}-\frac{\partial V}{\partial x_i} \right)^2 \sigma_i^2 \right]. \label{theta_change}
\end{align}

\noindent Given Equations \eqref{Fchange} and \eqref{G_change}, our Dirichlet problem becomes

\begin{align}\label{DP}
\frac{\partial V}{\partial \tau} + F\left(\tau,x,V,DV,D^2V \right) = 0  \quad  &\hbox{in} \quad \Omega\times\left[0,T\right] \nonumber\\
V\left(0,x_1,...,x_N \right) = V_0\left(x_1,...,x_N\right) \quad &\hbox{in} \quad \Omega 
\end{align}

\noindent where $V_0\left(x_1,...,x_N\right)$ is the initial condition. Hence, the main theorem of this work is defined as follows

\begin{thm}\label{main-th}
Assume that the differential matrix with respect to the Hessian matrix of the nonlinear operator $F$ is negative definite. Then, the problem \eqref{DP} has at least one viscosity solution.
\end{thm}

\noindent Before passing to the proof of the theorem, we are going to state some important definitions that will be used afterwards. Given an open set $\Omega_T \subset \mathbb{R}^{N+1}$, we recall that $V$ is \textit{lower semi-continuous} (LSC) or \textit{upper semi-continuous} (USC) at $\left(t,x\right)$ if for all sequences $\left(s_n, y_n \right) \rightarrow \left(t, x \right)$,

\begin{alignat}{2}
V\left(t,x\right) &\leq \liminf \limits_{n \rightarrow \infty}  V\left(s_n, y_n\right) \quad &&\text{(LSC)} \nonumber \\
V\left(t,x\right) &\geq \limsup \limits_{n \rightarrow \infty}  V\left(s_n, y_n\right) \quad &&\text{(USC)} \nonumber.
\end{alignat}

\noindent Moreover, we define $V_{*}$ the \textit{lower semi-continuous envelope of V} as the largest lower semi-continuous function lying below $V$ and $V^{*}$ the correspondent  \textit{upper semi-continuous envelope of V} as the smallest upper semi-continuous function lying above $V$.\\

\noindent Let us continue by presenting the definition of viscosity solutions, which are the type of solutions that we will look for. Let us recall  $\Omega_T = \left[0,T \right] \times \mathbb{R}^N$ and a function $V \in C^{1,2}\left(\Omega_T\right)$. Then, we have the following definitions.

\begin{defn}
\noindent $U$ is a subsolution of \eqref{DP} if $U$ is upper semi-continuous and if, for all $\left(t,x\right) \in \Omega_T$ and all the test functions $\phi$ such that $U \leq \phi$ in a neighbourhood of $\left(t,x\right)$ and $U\left(t,x\right)=\phi\left(t,x\right)$ , we have that

\begin{align}
\frac{\partial \phi}{\partial \tau} + F\left(\tau,x,\phi,D\phi,D^2\phi \right) \leq 0. 
\end{align}
 
\noindent $U$ is a supersolution of \eqref{DP} if $U$ is lower semi-continuous and if, for all $\left(t,x\right) \in \Omega_T$ and all the test functions $\phi$ such that $U \geq \phi$ in a neighbourhood of $\left(t,x\right)$ and $U\left(t,x\right)=\phi\left(t,x\right)$, we have that

\begin{align}
\frac{\partial \phi}{\partial \tau} + F\left(\tau,x,\phi,D\phi,D^2\phi \right) \geq 0. 
\end{align} 

\noindent Finally, $U$ is a solution of \eqref{DP} if it is both a sub and supersolution. 
 
\end{defn}

\noindent Now we can present Perron method to find a solution of problem \eqref{DP}. First of all, we require that the nonlinear operator $F$ is degenerate elliptic. Then, Perron method is defined as follows.

\begin{thm}\label{perron}
Assume $w$ is a subsolution of problem \eqref{DP} and $v$ is a supersolution of problem \eqref{DP} such that $w \leq v$. Suppose also that there is a subsolution $\underline{u}$ and a supersolution $\overline{u}$ of problem \eqref{DP} that satisfy the boundary condition $\underline{u}_{*}\left(t,x\right) = \overline{u}^{*}\left(t,x\right) = g\left(t,x\right)$. Then,

\begin{align}
W\left(t,x\right) = \sup \lbrace w\left(t,x\right): \underline{u} \leq w \leq \overline{u} \, \text{and} \, w \, \text{is a subsolution of \eqref{DP}} \rbrace.
\end{align}

\end{thm}

\noindent In order to apply the Perron method we first have to  set a subsolution and supersolution of problem \eqref{DP}. Then, we have to construct a maximal subsolution such that it lies between both sub and supersolutions. Finally, we have to define the proper comparison principle such that the boundary condition defined in Theorem \ref{perron} holds. 

\noindent Hence, let us start by recalling the equivalent "Black-Scholes" linear problem. If we denote the linear elliptic operator as

\begin{align}
\tilde{F}\left(\tau,x,V,DV,D^2V \right) = - \frac{1}{2} \sumai \sumaj \sigma_i \sigma_j \rho_{ij} \frac{\partial^2 V}{\partial x_i \partial x_j} - \sumai \frac{\partial V}{\partial x_i} \left(r-\frac{\sigma_i^2}{2}\right) + rV ,\label{Flin}
\end{align} 

\noindent then there exists a unique solution $\Lambda$ of the problem

\begin{align}
\frac{\partial V}{\partial \tau} + \tilde{F}\left(\tau,x,V,DV,D^2V \right) = 0  \quad  &\hbox{in} \quad \Omega\times\left[0,T\right] \nonumber\\
V\left(0,x_1,...,x_N \right)= V_0\left(x_1,...,x_N\right) \quad  &\hbox{in} \quad \Omega \label{BS}
\end{align}

\noindent Based on the existence of this unique solution $\Lambda$, we will construct our sub and supersolutions. Then, the following Lemma presents both sub and supersolutions of problem \eqref{DP}.

\begin{lem}\label{lemma-sub-super} 
Let $F$ be the nonlinear elliptic operator defined in Equation \eqref{Fchange}. Then the following functions are sub and supersolutions of problem \eqref{DP}.

\begin{align}
\overline{V} &= \Lambda + C\tau \nonumber\\
\underline{V} &= \Lambda - C\tau\nonumber
\end{align}

\noindent where $\Lambda$ is the unique solution of problem \eqref{BS} and $C$ is a positive constant such

\begin{align}\label{perron-cond}
C \geq \sup_{x \in \Omega} \, \lvert G\left(x,D^2 \Lambda \right) \rvert
\end{align}

\end{lem}

\begin{proof}
\noindent Let us see that the $\underline{V}$ is a subsolution of \eqref{DP}. Firstly, the upper semi-continuity of $\underline{V}$ follows from the continuity of the solution $\Lambda\left(\tau,x\right)$. Let us see that for all test functions $\phi$ such that $\underline{V} \leq \phi$ in a neighbourhood of $\left(\tau,x\right)$ and $\underline{V}\left(\tau,x\right)=\phi\left(\tau,x\right)$, it follows that $ \frac{\partial \phi}{\partial \tau} + F\left(\tau,x,\phi,D\phi,D^2\phi \right)$ is negative.

\noindent Let $\phi$ be a test function such that $\underline{V} \leq \phi$. Then, we have

\begin{align*}
\frac{\partial \phi}{\partial \tau}\left(\tau,x\right) &= \frac{\partial \underline{V}}{\partial \tau}\left(\tau,x\right) \\
D\phi\left(\tau,x\right) &= D\underline{V}\left(\tau,x\right) \\
D^2\phi\left(\tau,x\right) &\geq D^2\underline{V}\left(\tau,x\right)
\end{align*}

\noindent Now we use the condition of degenerate ellipticity of the operator $F$. This condition implies that

\begin{align*}
\frac{\partial \phi}{\partial \tau} + F\left(\tau,x,\phi,D\phi,D^2\phi \right) &\leq  \frac{\partial \underline{V}}{\partial \tau} + F\left(\tau,x,\underline{V},D\underline{V},D^2\underline{V} \right) \\
&\leq  G\left(\tau,D^2\Lambda\right) - C\\
&\leq 0 
\end{align*}

\noindent where the last inequality holds using Condition \ref{perron-cond}.

\noindent Let us now prove that $\overline{V}$ is in fact a supersolution. In this case, the lower semi-continuity follows from the continuity of the solution $\Lambda$. Let us see that for all test functions $\phi$ such that $\overline{V} \geq \phi$ in a neighbourhood of $\left(\tau,x\right)$ and $\overline{V}\left(\tau,x\right)=\phi\left(\tau,x\right)$, it follows that $ \frac{\partial \phi}{\partial \tau} + F\left(\tau,x,\phi,D\phi,D^2\phi \right)$ is positive.

\noindent Let $\phi$ be a test function such that $\underline{\Lambda} \geq \phi$. Then, we have that

\begin{align*}
\frac{\partial \phi}{\partial \tau}\left(\tau,x\right) &= \frac{\partial \overline{V}}{\partial \tau}\left(\tau,x\right) \\
D\phi\left(\tau,x\right) &= D\overline{V}\left(\tau,x\right) \\
D^2\phi\left(\tau,x\right) &\leq D^2\overline{V}\left(\tau,x\right)
\end{align*}

\noindent Now we use the condition of degenerate ellipticity of the operator $F$ and Condition \ref{perron-cond}. Both conditions imply that

\begin{align*}
\frac{\partial \phi}{\partial \tau} + F\left(\tau,x,\phi,D\phi,D^2\phi \right) &\geq  \frac{\partial \overline{V}}{\partial \tau} + F\left(\tau,x,\overline{V},D\overline{V},D^2\overline{V} \right) \\
&\geq G\left(x,D^2\Lambda\right) + C \\
&\geq 0.
\end{align*}

\noindent Then, $\overline{V}$ is a supersolution of problem \eqref{DP}.
\end{proof}

\begin{rem}
By definition, it remains valid that $\underline{V} \leq \overline{V}$.
\end{rem}

\begin{rem}
Given the "Black-Scholes" solution $\Lambda$, the nonlinear term $G\left(x,D^2\Lambda\right)$ is bounded for every $x$ in $\Omega$. Based on the construction of the replicant portfolio, it is observed that the transaction costs are proportional to the second derivatives of the option. Moreover, from the solution of the linear problem, we know that the second derivatives tend to zero when the prices are too small or too large. Then, in those scenarios, the replicant potfolio is almost not rebalanced so that a little amount of stocks are traded resulting on a small contribution of the transaction costs function.
\end{rem}

\noindent Following Lemma 2.3.15 from \cite{imbert2013introduction}, there exists a function $U$ such that $\underline{V} \leq U \leq \overline{V}$ and $U^*$ is a subsolution of \eqref{DP} and $U_*$ is a supersolution of \eqref{DP}. Then, to finally prove Theorem \ref{main-th}, we need to confirm that $U^*\left(\tau,S\right) = U_*\left(\tau,S\right)$. For this purpose, we will consider the comparison principle stated in \cite{imbert2013introduction}.

\begin{prop}[Comparison Principle]
If $u$ is a subsolution of problem \eqref{DP} and $v$ is a supersolution of problem \eqref{DP} in $\Omega_T$ and $u \leq v$ on the parabolic boundary $\partial_p \Omega_T$, then $u \leq v$ in $\Omega_T$. 
\end{prop}

\noindent Hence, our last Lemma is stated below:

\begin{lem}
Let $\underline{V}$ and $\overline{V}$ be the sub and supersolutions of problem \eqref{DP} and $U$ the function obtained by Lemma 2.3.15 from \cite{imbert2013introduction} such that $\underline{V} \leq U \leq \overline{V}$. Then, $U^*\left(\tau,S\right) = U_*\left(\tau,S\right)$. 

\end{lem}

\begin{proof}
Let us first observe that the inequality $U_{*} \leq U^{*}$ holds by definition of the semi-continuous envelopes. For the other inequality let us recall $\underline{V}$ and $\overline{V}$ defined in Lemma \ref{lemma-sub-super} and, using the continuity of the linear solution $\Lambda$, we have that

\begin{align*}
\left(\underline{V}\right)_{*} = \underline{V} = \left(\underline{V}\right)^{*} \\
\left(\overline{V}\right)_{*} = \overline{V} = \left(\overline{V}\right)^{*}
\end{align*}

\noindent In particular, in the parabolic boundary, we find that both sub and supersolutions are equal to $\Lambda$. Then, it is valid that

\begin{align}
\left(\overline{V}\right)^* \leq \left(\underline{V}\right)_* \, \, \text{in} \, \, \partial_p\Omega_T 
\end{align}

\noindent Moreover, for Lemma 2.3.15 from \cite{imbert2013introduction}, $\underline{V} \leq U \leq \overline{V}$. Using this result and the previous inequality, it follows that

\begin{align}
U^* \leq U_* \, \, \text{in} \, \, \partial_p\Omega_T 
\end{align}

\noindent Finally, the expected inequality is obtained by the comparison's principle result.
\end{proof}

\section{Numerical Implementation}
\subsection{Numerical Framework}

\noindent In this section we derive a numerical framework that is used to find an approximate solution of problem \eqref{problema_completo}. This solution will help us to understand how the presence of transaction costs affects the pricing of a specific financial option. With this aim, we develop an iterative scheme such that on every step, an approximate solution is found. Each step is then repeated until the convergence of the scheme. By recalling the nonlinear problem, we propose the following iterative process 

\begin{alignat}{2}
-U_{\tau}^n+\mathcal{L}U^n &= G\left(U^{n-1}\right) \quad &&\hbox{in} \quad \Omega\times\left[0,T\right] \nonumber\\
U^n\left(0,x_1,...,x_n\right) &= U_0\left(x_1,...,x_n\right) \quad &&\hbox{in} \quad \Omega \label{lin_prob_3}
\end{alignat}

\noindent with $U^0\left(\tau,x_1,...,x_n \right) = 0$, $\hbox{dim} \, \Omega=2$ and $U\left(\tau,x\right) = V\left(\tau,x\right)$ as defined in Equation \eqref{DP}. For numerical convenience, we approximate the original smooth domain by a discrete one $\hat{\Omega}_T\subset\left[ a,b \right] \times \left[ a,b \right] \times \left[ 0,T \right]$, setting $a$ and $b$ in order to cover a set of feasible stock prices. The step of the spatial variables is uniformly set as $\Delta x=\left(b-a \right)/S_x$, being $S_x$ the number of grid points in the x- direction. The step of the temporal variable is also uniformly set as $\Delta \tau=T/T_x$ being $T_x$ the number of grid points in the $\tau$- direction. We define $n$ as the step of the iterative problem and, given $n$, $m$ as each of the temporal steps. Hence, we define the solution to the $n$-step iterative problem as $U^{m}_{ij} = U\left(x_i,y_j,m\Delta\tau\right)$ where $0\leq i,j \leq S_x$ and $0\leq m \leq T_x$.\\

\noindent At each step $n$, we have to solve a linear problem involving both second and mixed derivatives of $U$. If we apply directly a finite difference scheme, the invertible matrix would not be tridiagonal as mixed spatial derivatives have to be considered. Hence, we apply an Alternating Direction Implicit (ADI) method with a Finite Difference approach (FD).\\

\noindent We follow the steps presented in the work of \cite{jeong2013comparison} to determine the two stages of the procedure. The main idea of the ADI method is to generate an intermediate step $m+1/2$ between steps $m$ and $m+1$. The first half step is taken implicitly in the x-direction and explicitly in the y-direction. The other half step is taken implicitly in the y-direction and explicitly in the x-direction. \\

\noindent In the first place, we split the temporal derivative as shown on (\ref{temp_der}) 

\begin{align}
U_{\tau} \simeq \frac{U^{m+1}_{ij}-U^{m}_{ij}}{\Delta t}=\frac{U^{m+1}_{ij}-U^{m+\frac{1}{2}}_{ij}}{\Delta t}+\frac{U^{m+\frac{1}{2}}_{ij}-U^{m}_{ij}}{\Delta t} \label{temp_der}.
\end{align} 

\noindent Then, we discretize the lineal operator

\begin{align*}
\mathcal{L}=\frac{1}{2} \sumai \sumaj \sigma_i \sigma_j \rho_{ij} \frac{\partial^2 U}{\partial x_i \partial x_j} + \sumai \frac{\partial U}{\partial x_i} \left( r-\frac{\sigma_i^2}{2} \right) -rU, 
\end{align*}

\noindent setting

\begin{align*}
\frac{\partial U}{\partial x_1}&\simeq\frac{U^m_{i+1,j}-U^m_{i,j}}{\Delta x},\\
\frac{\partial U}{\partial x_2}&\simeq\frac{U^m_{i,j+1}-U^m_{i,j}}{\Delta x},\\
\frac{\partial^2 U}{\partial x_1^2}&\simeq\frac{U^m_{i+1,j}-2U^m_{i,j}+U^m_{i-1,j}}{\Delta x^2},\\
\frac{\partial^2 U}{\partial x_2^2}&\simeq\frac{U^m_{i,j+1}-2U^m_{i,j}+U^m_{i,j-1}}{\Delta x^2},\\
\frac{\partial^2 U}{\partial x_1 x_2}&\simeq \frac{U^m_{i+1,j+1}+U^m_{i-1,j-1}-U^m_{i-1,j}-U^m_{i,j-1}}{4\Delta x^2}.
\end{align*}

\noindent As in section 2.1 of \cite{jeong2013comparison}, we split the discretization of the operator $\mathcal{L}$ between

\begin{align*}
\mathcal{L}^{x}&= \frac{\sigma_1^2}{4} \frac{U^{m+\frac{1}{2}}_{i+1,j}-2U^{m+\frac{1}{2}}_{i,j}+U^{m+\frac{1}{2}}_{i-1,j}}{\Delta x^2} + \frac{\sigma_2^2}{4} \frac{U^m_{i,j+1}-2U^m_{i,j}+U^m_{i,j-1}}{\Delta x^2} + \frac{1}{2} \sigma_1 \sigma_2 \rho  \frac{U^m_{i+1,j+1}+U^m_{i-1,j-1}-U^m_{i-1,j}-U^m_{i,j-1}}{4\Delta x^2} \nonumber \\ 
&+\frac{1}{2} \left(r-\frac{\sigma_1^2}{2} \right) \frac{U^{m+\frac{1}{2}}_{i+1,j}-U^{m+\frac{1}{2}}_{i,j}}{\Delta x} + \frac{1}{2} \left(r-\frac{\sigma_2^2}{2} \right) \frac{U^{m}_{i,j+1}-U^{m}_{i,j}}{\Delta x}-\frac{1}{2}rU^{m+\frac{1}{2}}_{ij}
\end{align*}

\noindent and

\begin{align*}
\mathcal{L}^{y}&= \frac{\sigma_1^2}{4} \frac{U^{m+\frac{1}{2}}_{i+1,j}-2U^{m+\frac{1}{2}}_{i,j}+U^{m+\frac{1}{2}}_{i-1,j}}{\Delta x^2} + \frac{\sigma_2^2}{4} \frac{U^{m+1}_{i,j+1}-2U^{m+1}_{i,j}+U^{m+1}_{i,j-1}}{\Delta x^2} + \frac{1}{2} \sigma_1 \sigma_2 \rho  \frac{U^{m+\frac{1}{2}}_{i+1,j+1}+U^{m+\frac{1}{2}}_{i-1,j-1}-U^{m+\frac{1}{2}}_{i-1,j}-U^{m+\frac{1}{2}}_{i,j-1}}{4\Delta x^2} \nonumber  \\ 
&+\frac{1}{2} \left(r-\frac{\sigma_1^2}{2} \right) \frac{U^{m+\frac{1}{2}}_{i+1,j}-U^{m+\frac{1}{2}}_{i,j}}{\Delta x} + \frac{1}{2} \left(r-\frac{\sigma_2^2}{2} \right) \frac{U^{m+1}_{i,j+1}-U^{m+1}_{i,j}}{\Delta x}-\frac{1}{2}rU^{m+1}_{ij}
\end{align*}

\noindent obtaining a two-stage full scheme

\begin{align*}
\frac{U^{m+\frac{1}{2}}_{ij}-U^{m}_{ij}}{\Delta t}=\mathcal{L}^x U_{ij}^{m+\frac{1}{2}},\\
\frac{U^{m+1}_{ij}-U^{m+\frac{1}{2}}_{ij}}{\Delta t}=\mathcal{L}^y U_{ij}^{m+1}.
\end{align*}
 
\noindent As the problem \eqref{lin_prob_3} contains the linear function $G$, we add this term on the second stage of the procedure by redefining $\tilde{\mathcal{L}}^y=\mathcal{L}^y-G$.

\begin{align*}
\frac{U^{m+1}_{ij}-U^{m}_{ij}}{\Delta t}=\mathcal{L}^x U_{ij}^{m+\frac{1}{2}}+\mathcal{L}^y U_{ij}^{m+1}-G\left( \cdot\right)=\mathcal{L}^x U_{ij}^{m+\frac{1}{2}} + \tilde{\mathcal{L}}^y U_{ij}^{m+1}.
\end{align*} 
 
\noindent The proposed framework is used to calculate first $U^{m+\frac{1}{2}}$ and then $U^{m+1}$. The most important gain with the ADI method is that it only requires the solution of two tridiagonal sets of equations at each time step. 
 
\subsection{Numerical Results} 

\noindent In order to implement the framework proposed in section \ref{demostracion}, we select a type of multi-asset option and a transaction costs function. First, we price a best cash-or-nothing option call on two assets. This option pays out a predefined cash amount $K$ if assets $S_1$ or $S_2$ are above or equal to the strike price $X$. The closed-form formula is presented on \cite{haug2007complete} as

\begin{align}
c_{best}=Ke^{-rT}\left[M\right.&\left.\left(y,z_1;-\rho_1\right)+M\left(-y,z_2;-\rho_2\right) \right] \\ \label{cbest}
y=\frac{\ln\left( S_1/S_2\right)+\frac{\sigma^2}{2}T}{\sigma \sqrt{T}}&, \quad \quad \sigma=\sqrt{\sigma_1^2+\sigma_2^2-2\sigma_1\sigma_2 \rho}\nonumber\\
z_1=\frac{\ln\left( S_1/X\right)+\frac{\sigma_1^2}{2}T}{\sigma_1 \sqrt{T}}&, \quad \quad z_2=\frac{\ln\left( S_2/X\right)+\frac{\sigma_2^2}{2}T}{\sigma_2 \sqrt{T}} \nonumber\\
\rho_1=\frac{\sigma_1-\rho}{\sigma}&, \quad \quad \rho_2=\frac{\sigma_2-\rho}{\sigma} \nonumber 
\end{align}

\noindent where $S_1$ and $S_2$ are the stock prices, $\sigma_1$ and $\sigma_2$ are the volatilities, $\rho$ is the correlation between both assets, $T$ is the maturity and $M\left(a,b;\rho\right)$ is

$$
M\left(a,b;\rho\right)=\frac{1}{2\pi\sqrt{1-\rho^2}} \int_{-\infty}^{a} \int_{-\infty}^{b} \exp \left[-\frac{x^2+y^2-2\rho xy}{2\left(1-\rho^2 \right)} \right] dx \, dy.
$$

\noindent Second, given the nonlinear function $G$ defined on \eqref{G_change}, we choose an exponential decreasing transaction costs function defined as

\begin{align*}
C\left(x\right)=C_0 \,e^{-\tilde{k} \, x}
\end{align*}

\noindent for each asset $x$. Hence, by recalling \eqref{rtc}, we can see that

\begin{align*}
E\left[ C\left(\sqrt{\Delta t} \, \left| \Phi_i \right| \right) \left| \Phi_i \right| \right]  & =  \int_{0}^{+\infty} C_0 \, e^{-\tilde{k} \sqrt{\Delta t} x } \, \frac{2 \, x}{\sqrt{2\pi \Theta_i}} e^{-x^2/2\Theta_i} \, dx \\
& =  C_0\sqrt{\frac{2}{\pi}} \, \int_{0}^{+\infty} e^{-\tilde{k} \sqrt{\Delta t} x} \, \frac{x}{\sqrt{\Theta_i}} \, e^{-x^2/2\Theta_i} \, dx \\ 
&= C_0\sqrt{\frac{2}{\pi}} \, \int_{0}^{+\infty} e^{-\tilde{k} \sqrt{\Delta t \Theta_i} y} \,\sqrt{\Theta_i} y \, e^{-y^2/2} \, dy\\
& =  C_0 \sqrt{\Theta_i} \sqrt{\frac{2}{\pi}} \left[ 1 - e^{\tilde{k}^2 \Delta t \Theta_i /2} \, \tilde{k} \sqrt{\Delta t \Theta_i} \, \hbox{ERFC}\left(\tilde{k}\sqrt{\frac{\Delta t \Theta_i}{2}} \right) \right].
\end{align*}

\noindent Then,

\begin{align}
G\left(S,D^2V \right)= C_0 \, \sqrt{\frac{2}{\pi}}\sum_{i=1}^2 \frac{S_i}{\Delta t} \sqrt{\Theta_i}  \left[ 1 - e^{\tilde{k}^2 \Delta t \Theta_i /2} \, \tilde{k} \sqrt{\Delta t \Theta_i} \, \hbox{ERFC}\left(\tilde{k}\sqrt{\frac{\Delta t \Theta_i}{2}} \right) \right]. 
\end{align}\label{exponentialF}

\noindent Table \ref{Table-parameters} presents the parameters chosen for the numerical implementation. Three different tests were then applied by varying the values of the stocks price, volatility, interest rate and strike among others. \\

\begin{table}[hbtp]
\centering
\begin{tabular}{lllllll}
\multirow{2}{*}{ Parameters} & \multicolumn{2}{c}{Testing 1} & \multicolumn{2}{c}{Testing 2} & \multicolumn{2}{c}{Testing 3} \\
&  Asset 1  & Asset 2 & Asset 1 & Asset 2 & Asset 1 & Asset 2\\ \midrule
$\sigma$ & $0.30$ & $0.15$ & $0.05$ & $0.1$ & $0.2$ & $0.2$  \\ 
$\rho$ & \multicolumn{2}{c}{$0.5$} & \multicolumn{2}{c}{$-0.3$} & \multicolumn{2}{c}{$0.2$} \\
$r$ &   \multicolumn{2}{c}{$0.08$} & \multicolumn{2}{c}{$0.02$} & \multicolumn{2}{c}{$0.1$}   \\ 
$T$ &  \multicolumn{2}{c}{$1$ year} & \multicolumn{2}{c}{$1$ year} & \multicolumn{2}{c}{$1$ year} \\ 
$K$ & \multicolumn{2}{c}{$5$}     & \multicolumn{2}{c}{$8$}  & \multicolumn{2}{c}{$6$}  \\ 
$X$ & \multicolumn{2}{c}{$30$}     & \multicolumn{2}{c}{$40$} & \multicolumn{2}{c}{$15$}   \\  \midrule
$\Delta x$ & \multicolumn{2}{c}{$1$} & \multicolumn{2}{c}{$1$} & \multicolumn{2}{c}{$1$} \\ 
$\Delta t_{TC} $ & \multicolumn{2}{c}{$1/261$}  & \multicolumn{2}{c}{$1/261$} & \multicolumn{2}{c}{$1/261$} \\ \midrule
$C_0$ &\multicolumn{2}{c}{$0.005$} & \multicolumn{2}{c}{$0.001$} & \multicolumn{2}{c}{$0.003$} \\ 
$\tilde{k}$ & \multicolumn{2}{c}{$1$} & \multicolumn{2}{c}{$0.5$} & \multicolumn{2}{c}{$0.7$} \\ \midrule
\end{tabular}
\caption{Numerical implementation parameters}
\label{Table-parameters}
\end{table}

\noindent We analyze different aspects of the ADI algorithm implemented and the dynamics of the general transaction cost model proposed. We focus on the following three points:

\begin{itemize}
\item[1.] Measure the impact of transaction costs in the option price.
\item[2.] Given an optimal number of iterations such that convergence is achieved, analyze the sensitivity of the final output to the choice of $\Delta t_{TC}$.
\item[3.] Given the iteration procedure proposed in \eqref{lin_prob_3}, determine the optimal number of $n$ such that the convergence is achieved and how the error diminishes as more steps are added.
\end{itemize}

\subsubsection{Transaction Costs impact}

\noindent Figure \eqref{tcimpact1} to Figure \eqref{tcimpact3} present the results for both transaction costs function and option price with transaction costs at $t=0$. By recalling the transaction costs function $G$ in \eqref{exponentialF} it can be noted that the costs are proportional to the assets spot price, the size of the second derivatives (i.e. Gamma) of the option price and the volatilities of each asset.\\

\noindent Testing 1 is defined based on a strike at $X=30$ with a premium paid at $K=5$. Figure \eqref{tcimpact1a} shows an exponential transaction costs function where the maximum is reached around the strike value. This behavior is expected as the maximum of Gamma is found near the at-the-money price. As these derivatives converge to zero when deep out-of-the-money or in-the-money, the transactions costs function vanishes. Figure \eqref{tcimpact1b} describes the dynamics of the option price when considering the transaction costs function.\\

\begin{figure}[hbtp]
\centering
\begin{minipage}[c]{0.45\textwidth}
\includegraphics[scale=0.55]{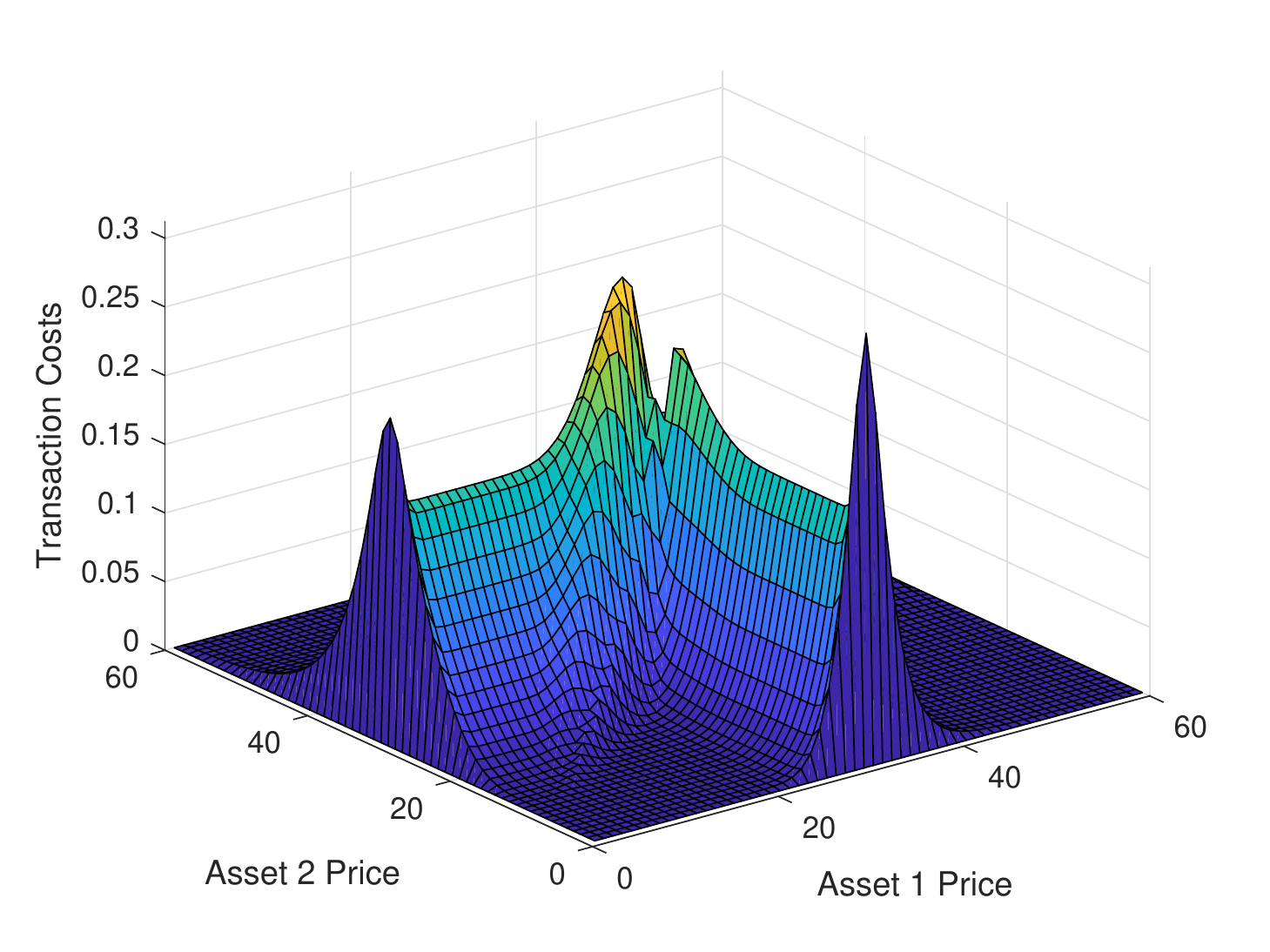}
	\subcaption{Transaction costs at time $t=0$.}\label{tcimpact1a}
\end{minipage}
\begin{minipage}[c]{0.45\textwidth}
\includegraphics[scale=0.55]{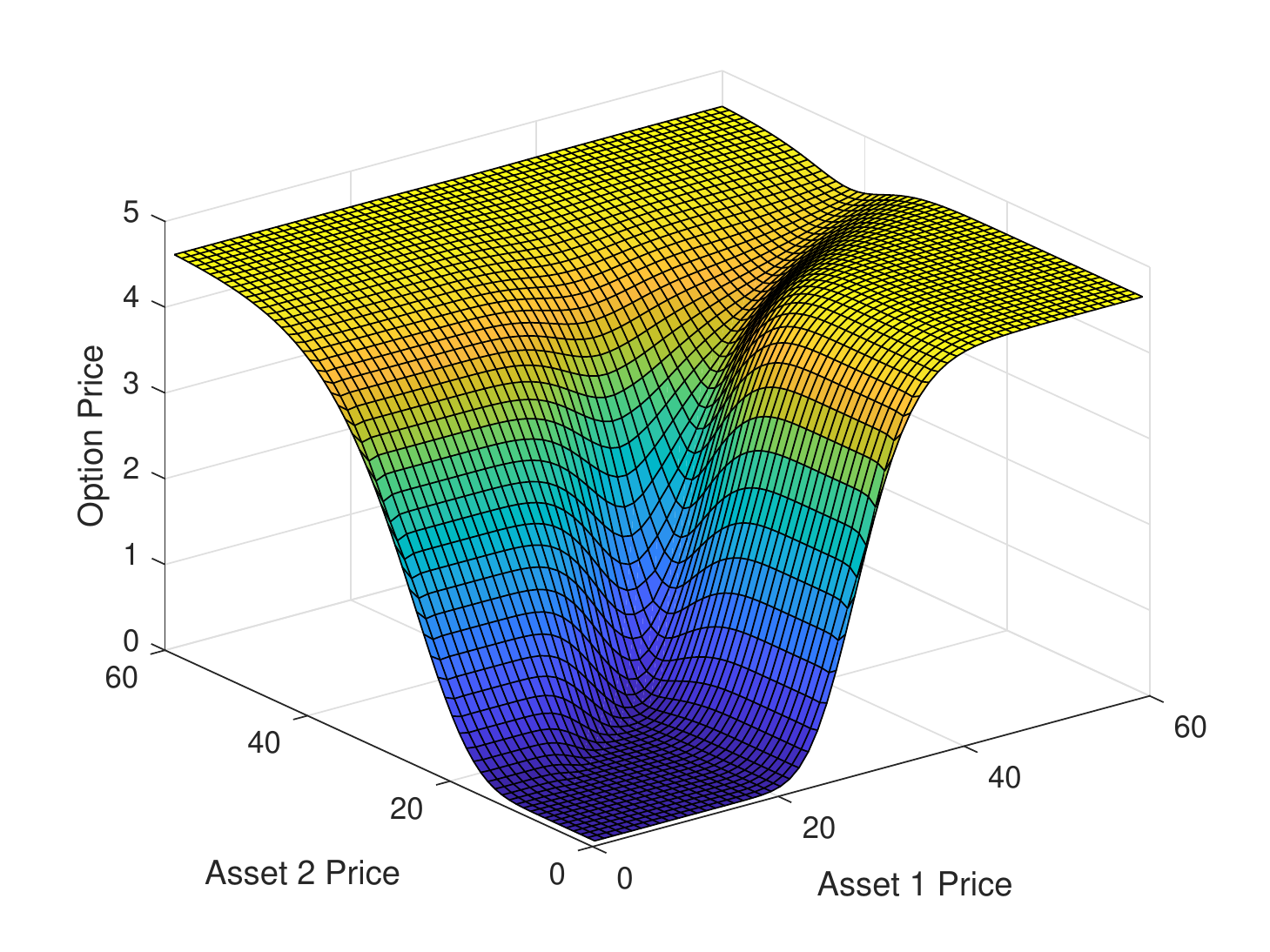}
	\subcaption{Option price at time $t=0$.}\label{tcimpact1b}
\end{minipage}
\caption{Testing 1.}\label{tcimpact1}
\end{figure}

\noindent The results for Testing 2 framework are presented in Figure \eqref{tcimpact2}. It is defined a strike value at $X=40$, a premium paid at $K=8$ and with two low volatile assets. The transaction costs function presented in Figure \eqref{tcimpact2a} shows a similar increasing pattern on its value up to the at-the-money region. Moreover, the higher volatilty of Asset 2 is observed by noting that transaction costs are higher when fixing a price for Asset 2 in comparison with Asset 1. As the option gets out-of-the-money, the shape of the transaction costs function becomes more symmetric and smoother.\\

\begin{figure}[hbtp]
\centering
\begin{minipage}[c]{0.45\textwidth}
\includegraphics[scale=0.55]{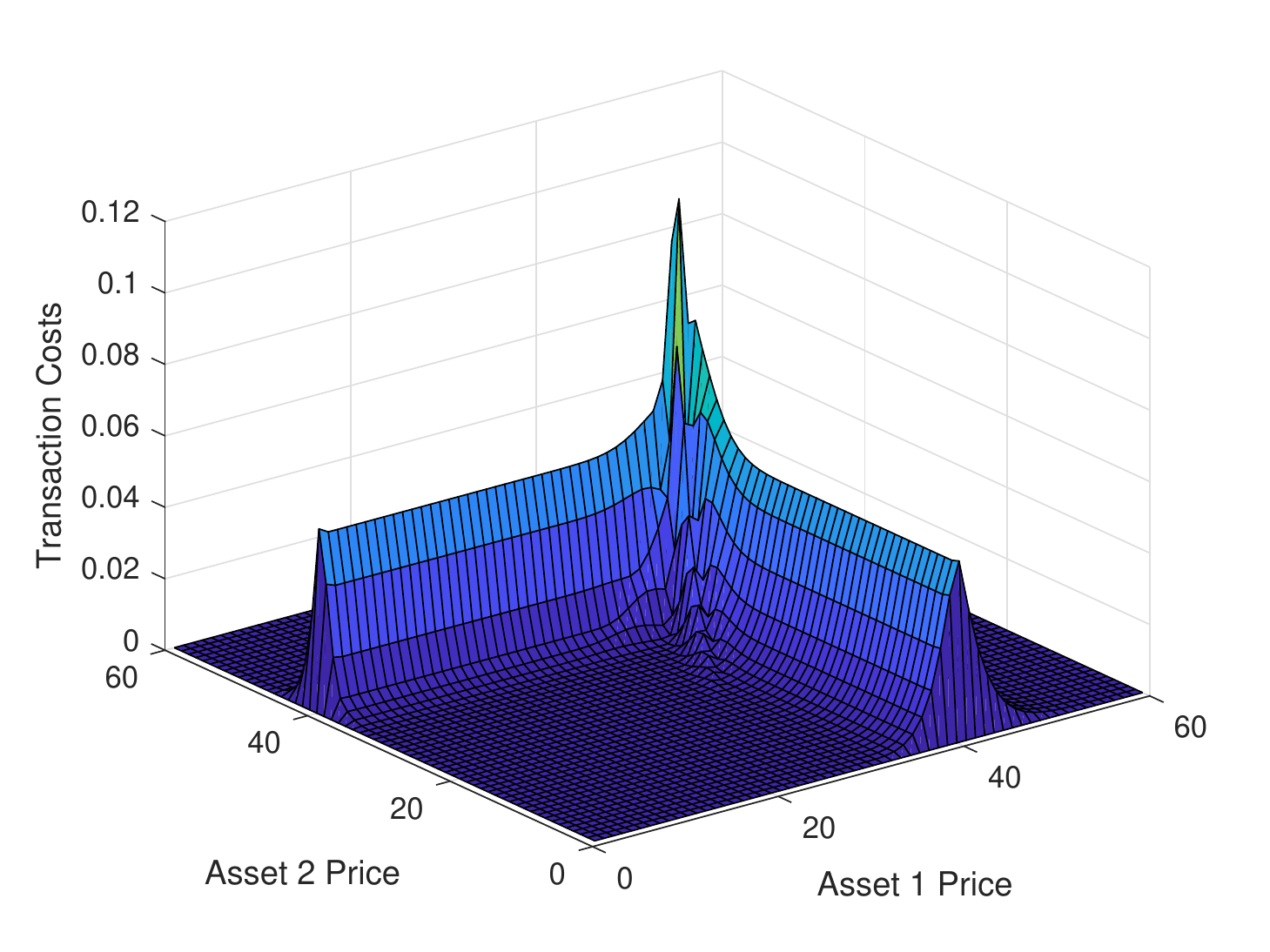}
\subcaption{Transaction costs at time $t=0$.}\label{tcimpact2a}
\end{minipage}
\begin{minipage}[c]{0.45\textwidth}
\includegraphics[scale=0.55]{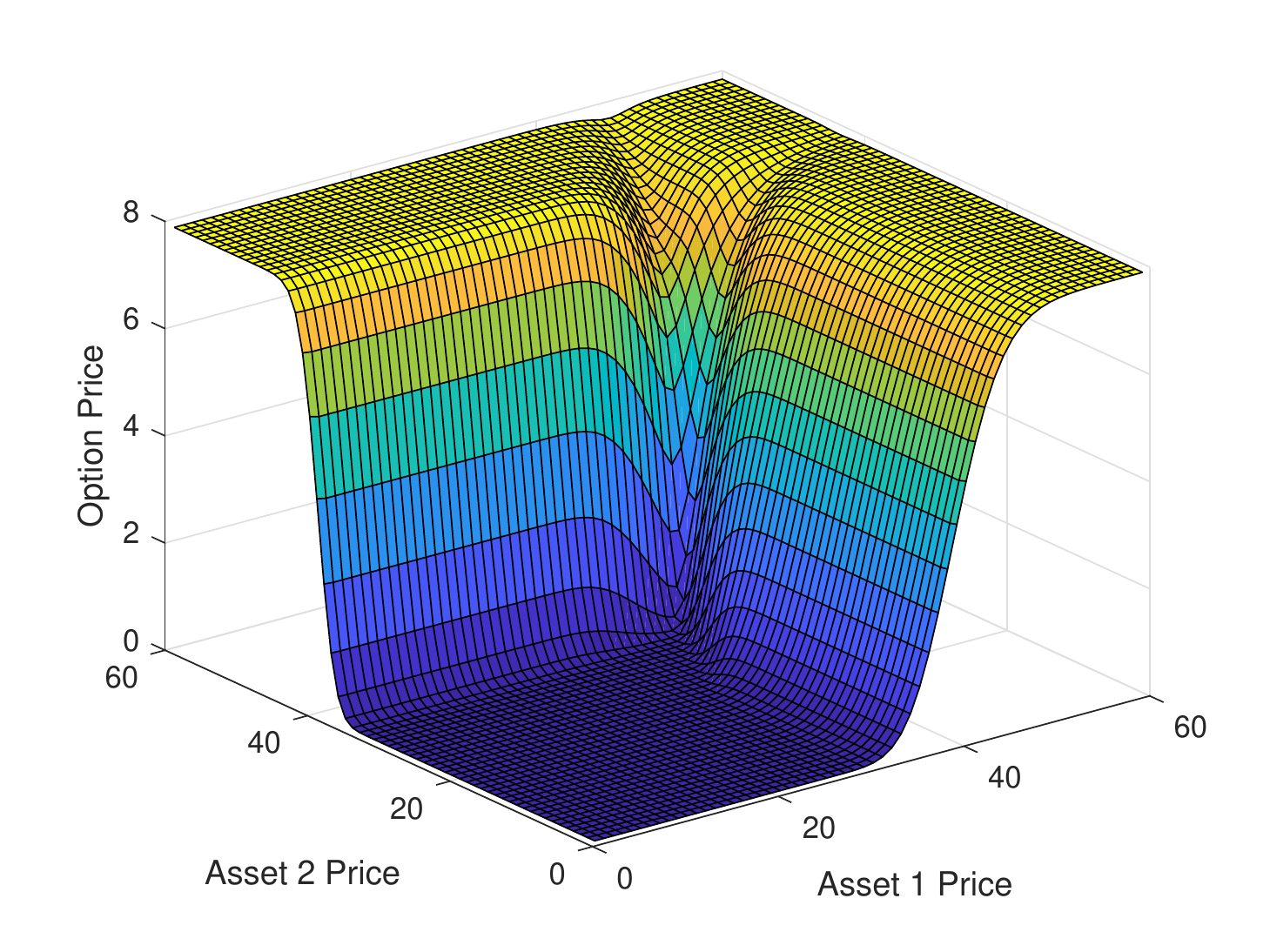}
	\subcaption{Option price at time $t=0$.}\label{tcimpact2b}
\end{minipage}
\caption{Testing 2.}\label{tcimpact2}
\end{figure}

\noindent Testing framework 3 is presented on Figure \eqref{tcimpact3}. Both assets are defined to have the same volatility but almost uncorrelated. The strike price is fixed at $X=15$ and the premium paid is equal to $K=6$. The symmetry observed in both Figures \eqref{tcimpact3a} and \eqref{tcimpact3b} are expected due to the design of the testing. Again, the maximum of the transaction costs function is reached when the prices are near the strike value and the converge to zero is seen when the option is deeper out-of-the-money or in-the-money. The option prices reflect the complementary pattern by showing a decrease in its value when the option is near the strike price.\\

\begin{figure}[hbtp]
\centering
\begin{minipage}[c]{0.45\textwidth}
\includegraphics[scale=0.55]{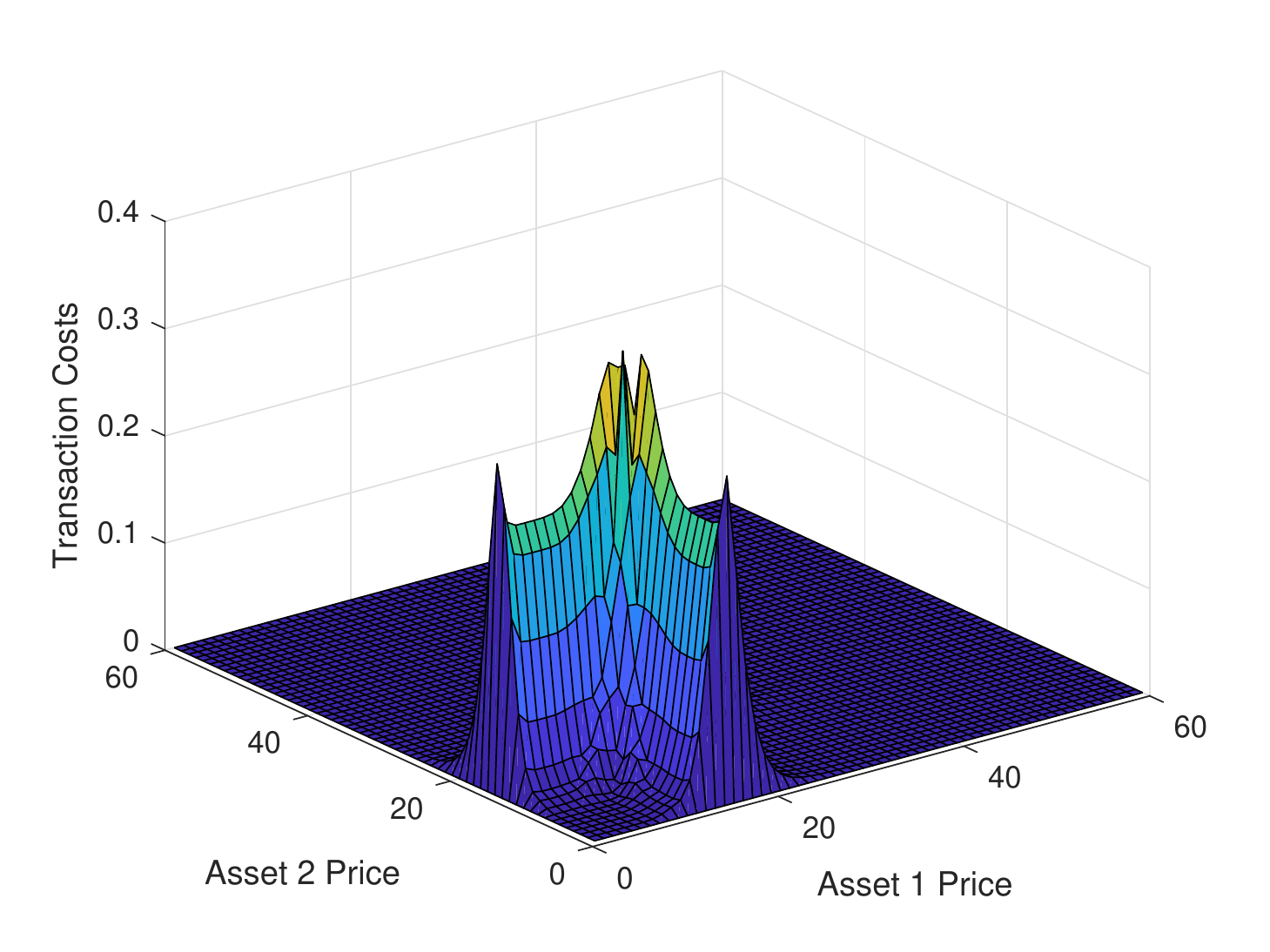}
	\subcaption{Transaction costs at time $t=0$.}\label{tcimpact3a}
\end{minipage}
\begin{minipage}[c]{0.45\textwidth}
\includegraphics[scale=0.55]{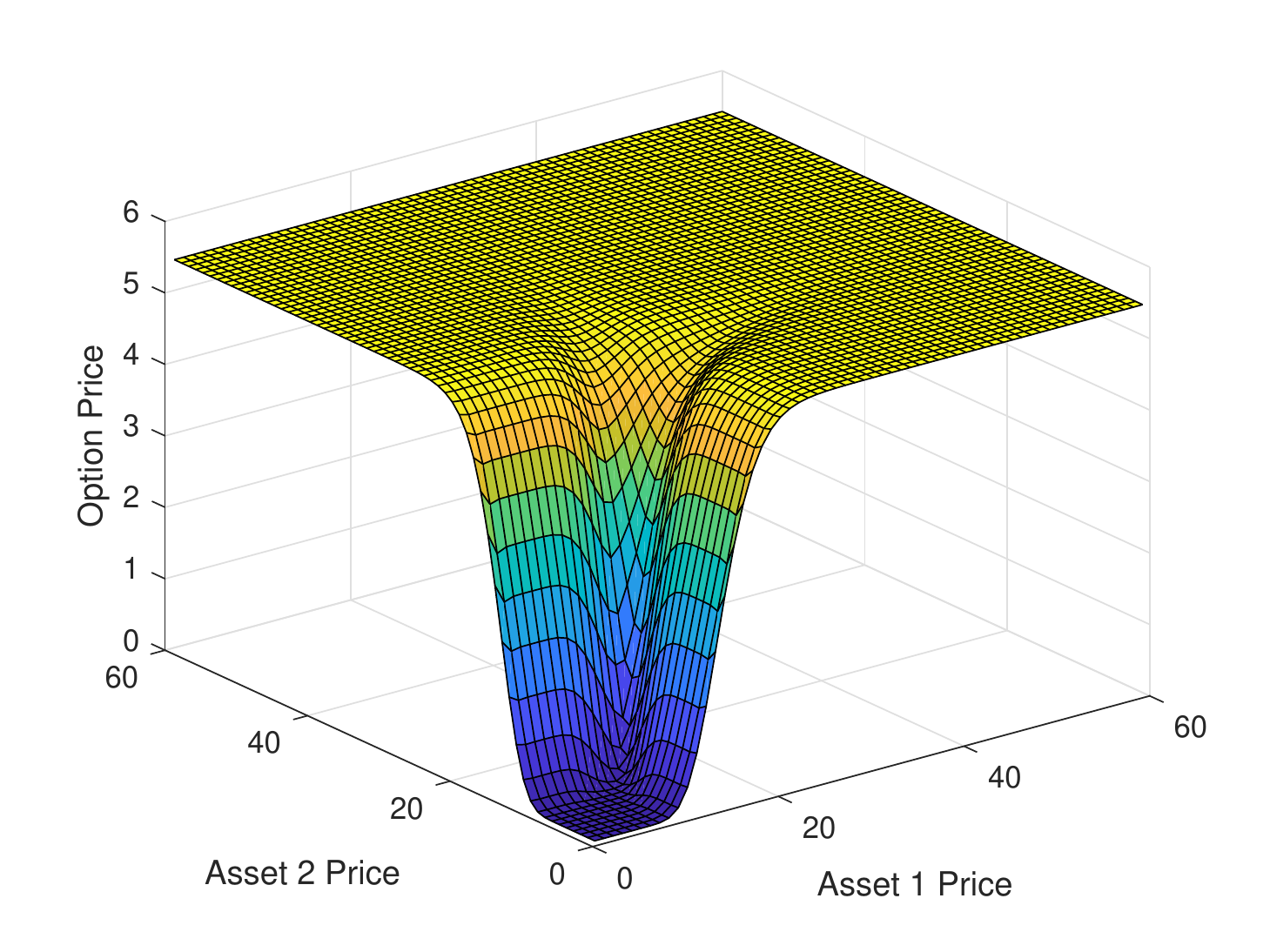}
	\subcaption{Option price at time $t=0$.}\label{tcimpact3b}
\end{minipage}
\caption{Testing 3.}\label{tcimpact3}
\end{figure}

\subsubsection{Sensitivity of the option to changes in $\Delta t_{TC}$}

\noindent In this section we study the sensitivity of the option price to changes in the size of the time-step $\Delta t_{TC}$ for rebalancing the replicant portfolio. By observing Equation \eqref{exponentialF}, it can be seen that the transaction costs function tends to infinity if $\Delta_{TC}$ tends to zero. Hence, we expect to see this results in the numerical testing. For this purpose, we ran Testing 2 framework under $100$ possible values of $\Delta_{TC}$ ranging from $7.6E-05$ (approximately rebalancing every $29$ minutes) to $0.007$ (approximately rebalancing every $2$ years). The results can be observed in Figures \ref{tau1} and \ref{tau2}. \\

\noindent Figure \ref{tau1} presents two different plots which show two states of the option price. In the panel on the left side it can be observed how the transaction costs behave when the price of Asset 1 is equal to $S_1 = 15$ and the parameter $\tau = 0$. It can be noted that the maximum value is reached at-the-money with a transaction cost of almost $2$. This maximum is also reached when $\Delta_{TC}$ is minimum. When the option becomes deeper in-the-money and out-of-the-money and $\Delta_{TC}$ increases, transaction costs tend to zero. A similar pattern is observed in the figure on the right side. The main difference resides on how the transaction costs highly increase as the Asset 1 price is set as of $S_1 = 40$. As Gamma is maximum near the at-the-money moneyness of the option, transaction costs explode near this pricing area. As it can be seen, the costs are of $34$ when both prices are set as $40$. This will be the case in which rebalancing is done too often so that the option price becomes negative due to the high amount of transaction costs payed.\\

\noindent The plot of the left side of Figure \ref{tau2} shows  the same dynamics for the case when Asset 1 price is equal to $S_1 = 55$. These dynamics are similar to the one observed in the plot of the left side of Figure \ref{tau1}. As Gamma decreases when prices are to low or to high, transaction costs  present the usual spike near the strike value. Moreover, this costs tend to zero as $\Delta_{TC}$ becomes larger and rebalancing is done less periodically. The plot on the right helps us to understand how far can the transaction costs increase. For this purpose, we fixed the prices of both assets as of $S_1 = 40$ and $S_2 = 40$. Then, we plotted the value of the transaction costs with respect to its time to maturity and the size of $\Delta_{TC}$. It can be observed that when $t=T$, transaction costs are near to zero as the option price is the predefined payoff. As time passes and the payoff is discounted, transaction costs increases as Gamma increases. When we reach $t=0$, transaction costs grow up to $34$. This result helps us to confirm the following expected conclusion: As the frequency of rebalancing of the replicant portfolio increases, transaction costs increase such that after a certain point of time, the option price turns into negative and the model becomes ill-posed.  

\begin{figure}[hbtp]
\centering
\begin{minipage}[t]{.5\linewidth}
\includegraphics[scale=0.6]{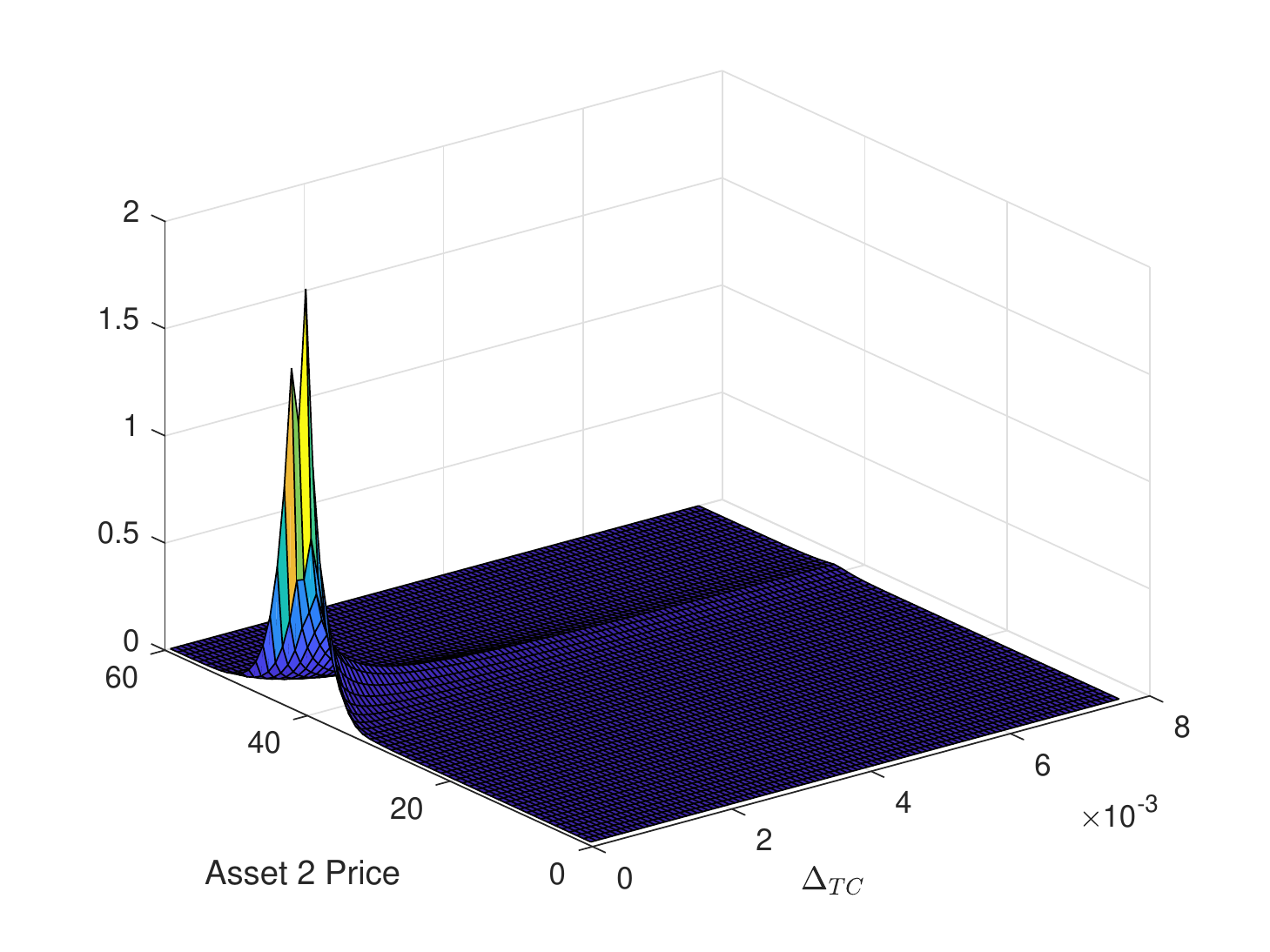}
\centering
\end{minipage}%
\begin{minipage}[t]{.5\linewidth}
\includegraphics[scale=0.6]{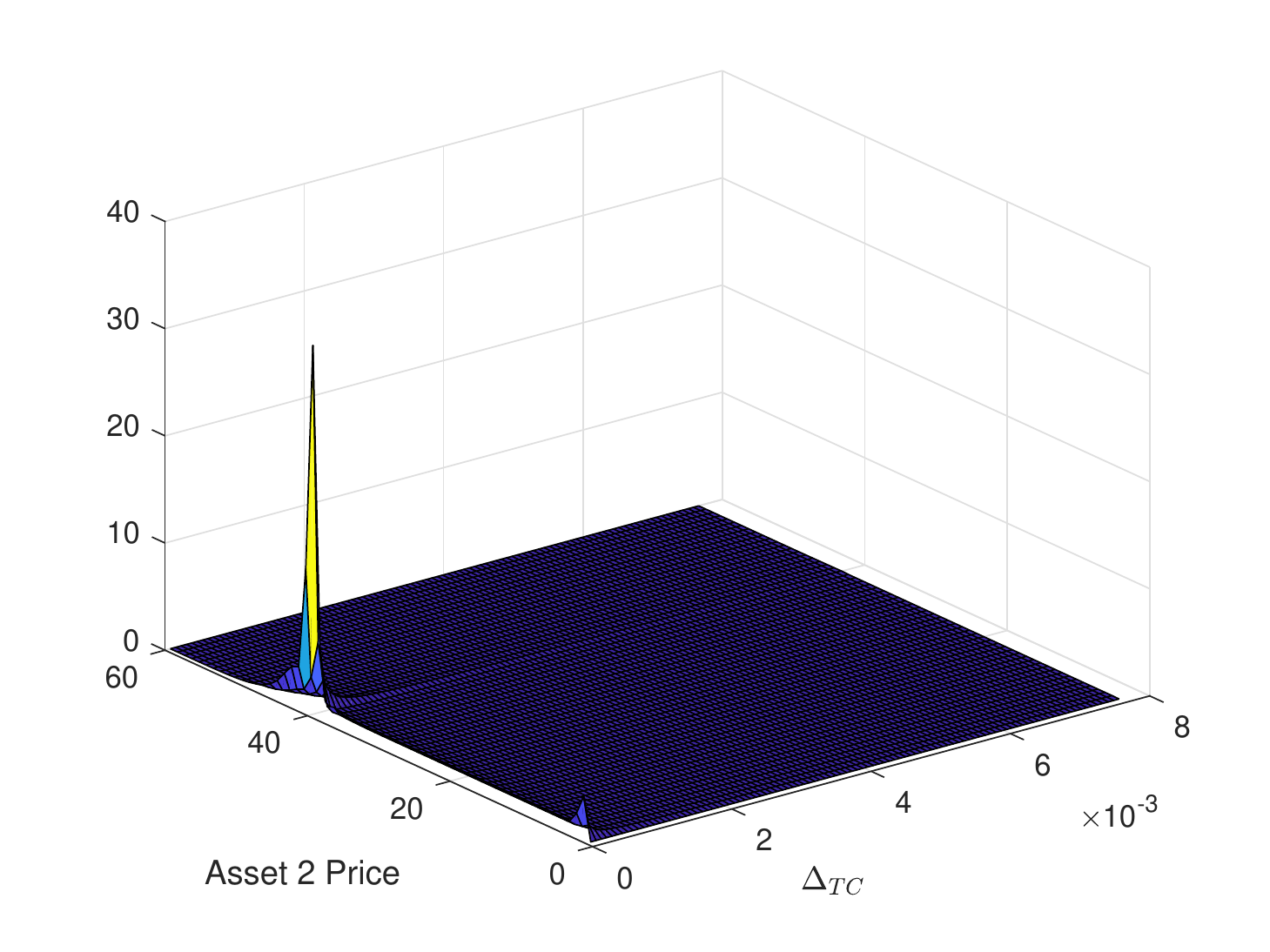}
\centering
\end{minipage}
\caption{Testing 2 - The figure on the left shows the transaction costs function at time $\tau = 0$ when Asset 1 price is equal to $15$. The figure on the right shows the transaction costs function at time $\tau = 0$ when Asset 1 price is equal to $40$.} \label{tau1}
\end{figure}

\begin{figure}[hbtp]
\centering
\begin{minipage}[t]{.5\linewidth}
\includegraphics[scale=0.6]{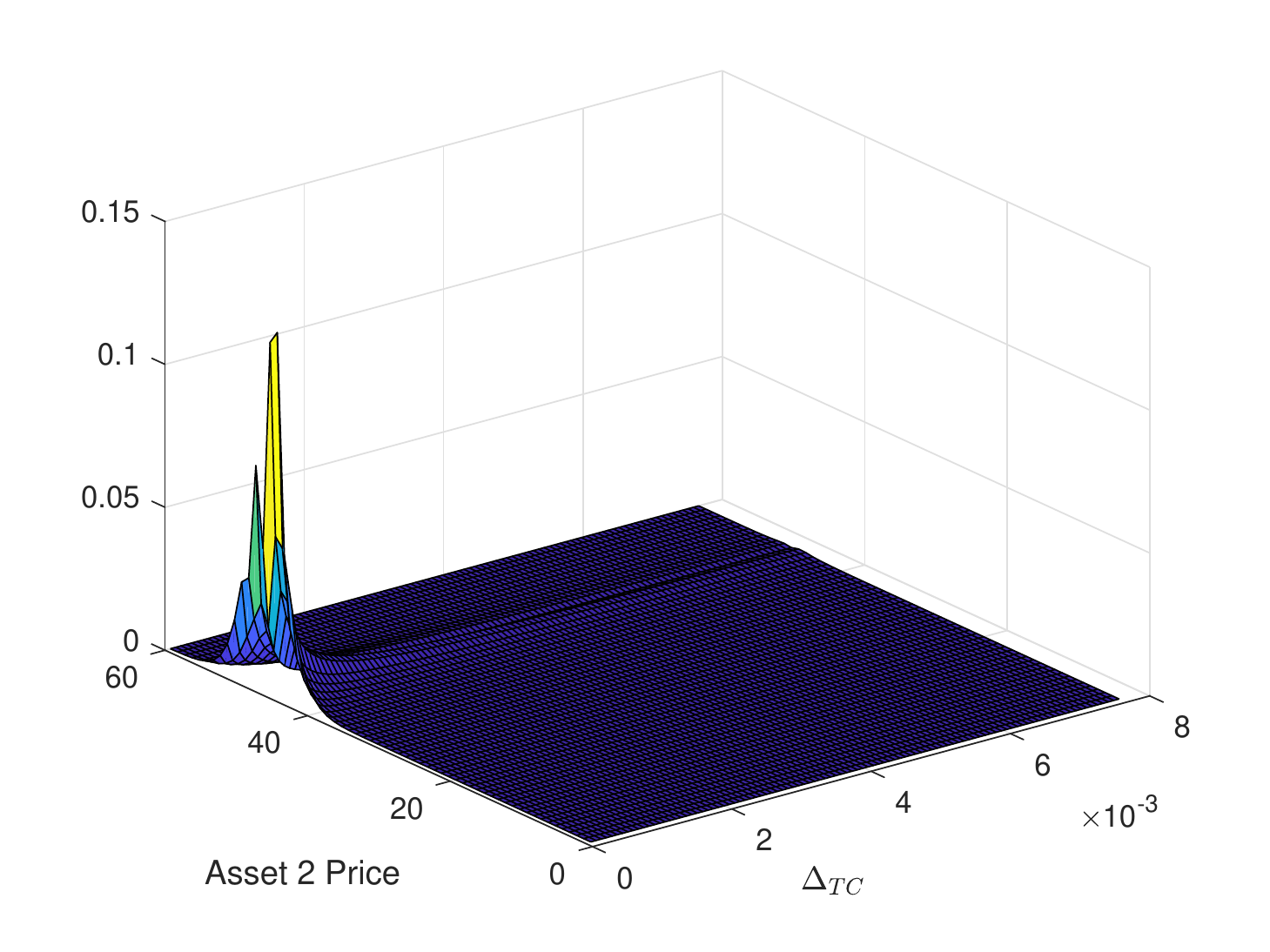}
\centering
\end{minipage}%
\begin{minipage}[t]{.5\linewidth}
\includegraphics[scale=0.6]{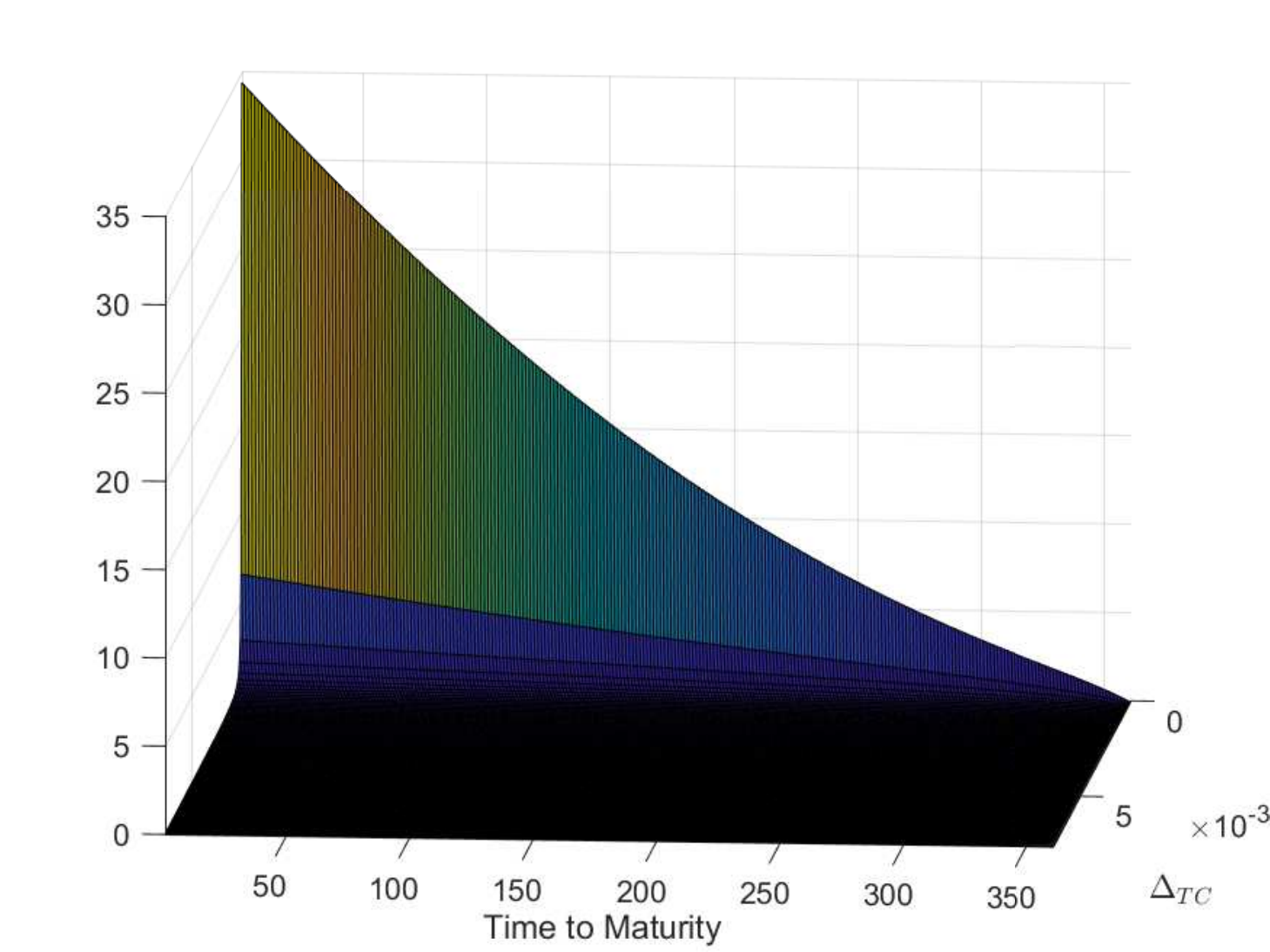}
\centering
\end{minipage}
\caption{Testing 2 - The figure on the left shows the transaction costs function at time $\tau = 0$ when Asset 1 price is equal to $55$. The figure on the right shows how the transaction costs function explodes when the option is at-the-money.}\label{tau2}
\end{figure}

\subsubsection{Convergence Analysis}

\noindent The third item of the previous list involves measuring the convergence of the iterative framework in terms of the differences between consecutive solutions. Our approach will follow from the observation that the result of each iteration correspond to a square matrix. Hence, fixing the last time step $\tau = T$, we calculate the distance between two consecutive final results. For this purpose, we use three different p-norms matrix which are: the $1$ norm, the $2$ norm and the $\infty$ norm. In summary, for each step $n$ and solution $U^n$, we calculate

\begin{align}
d\left(U^n,U^{n-1}\right) = \lVert U^n - U^{n-1} \rVert_p.
\end{align} 

\noindent Our objective is to see that this distance tends to zero as $n$ increases. Figure \eqref{convergence_testing1} to \eqref{convergence_testing3} present the plots of the results for the three scenarios. The figures plot the distance between two consecutive solutions against the iteration step $n+1$. In the right side, we provide a table with all the numerical results up to iteration $n=11$. In the first case, it can be seen that the three norms exponentially decrease to zero and, between iteration $7$ and $8$, convergence is achieved. In the second case, convergence is achieved even faster as by step $2$, the distance between both consecutive results is of order $E-05$. The third case is similar as the first scenario by noting that convergence is achieved at iteration $5$ with a distance between consecutive solutions of order $E-04$.



\begin{figure}[hbtp] 
\begin{subfigure}{0.48\textwidth}\centering`
\includegraphics[scale=0.55]{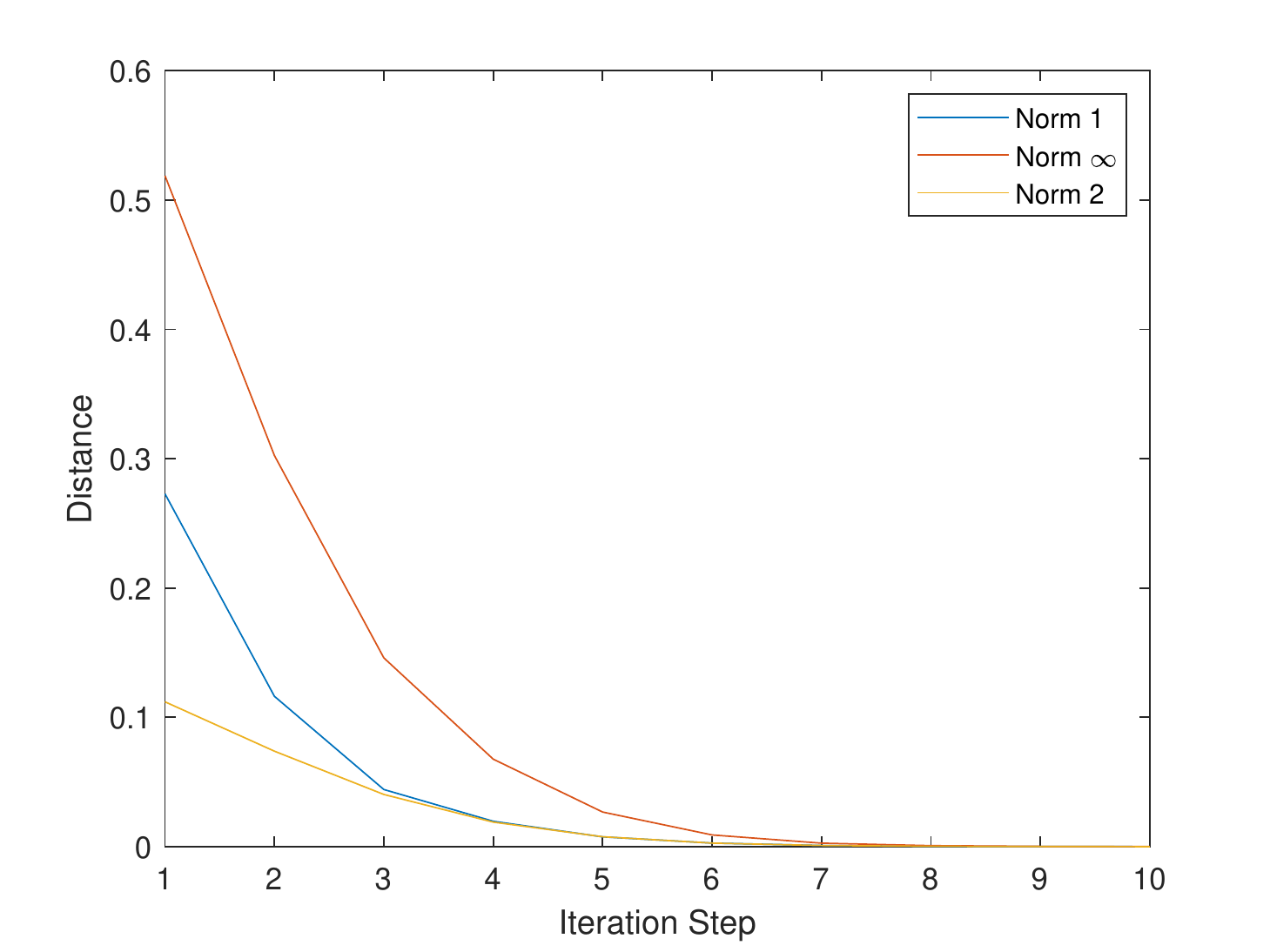}
\end{subfigure}\hspace*{\fill}
\begin{subfigure}[t]{0.48\textwidth}\centering
\footnotesize
\begin{tabular}{llll}
Iteration & Norm 1 & Norm 2 & Norm $\infty$ \\
\hline
1 & 0.2727 & 0.1120 & 0.5187\\
2 & 0.1162 & 0.0738 & 0.3024\\
3 & 0.0442 & 0.0403 & 0.1460\\
4 & 0.0196 & 0.0189 & 0.0676\\
5 & 0.0076 & 0.0076 & 0.0267\\
6 & 0.0028 & 0.0027 & 0.0090\\
7 & 8.8E-4 & 8.6E-4 & 0.0027\\
8 & 2.5E-4 & 2.5E-4 & 7.5E-4\\
9 & 8.1E-5 & 6.7E-5 & 1.9E-4\\
10 & 2.3E-5 & 1.6E-5 & 4.5E-5\\
\hline
\end{tabular}
\end{subfigure}
\caption{Convergence Analysis - Testing 1}\label{convergence_testing1}
\end{figure}
\begin{figure}[hbtp]
\begin{subfigure}{0.48\textwidth}\centering
\includegraphics[scale=0.55]{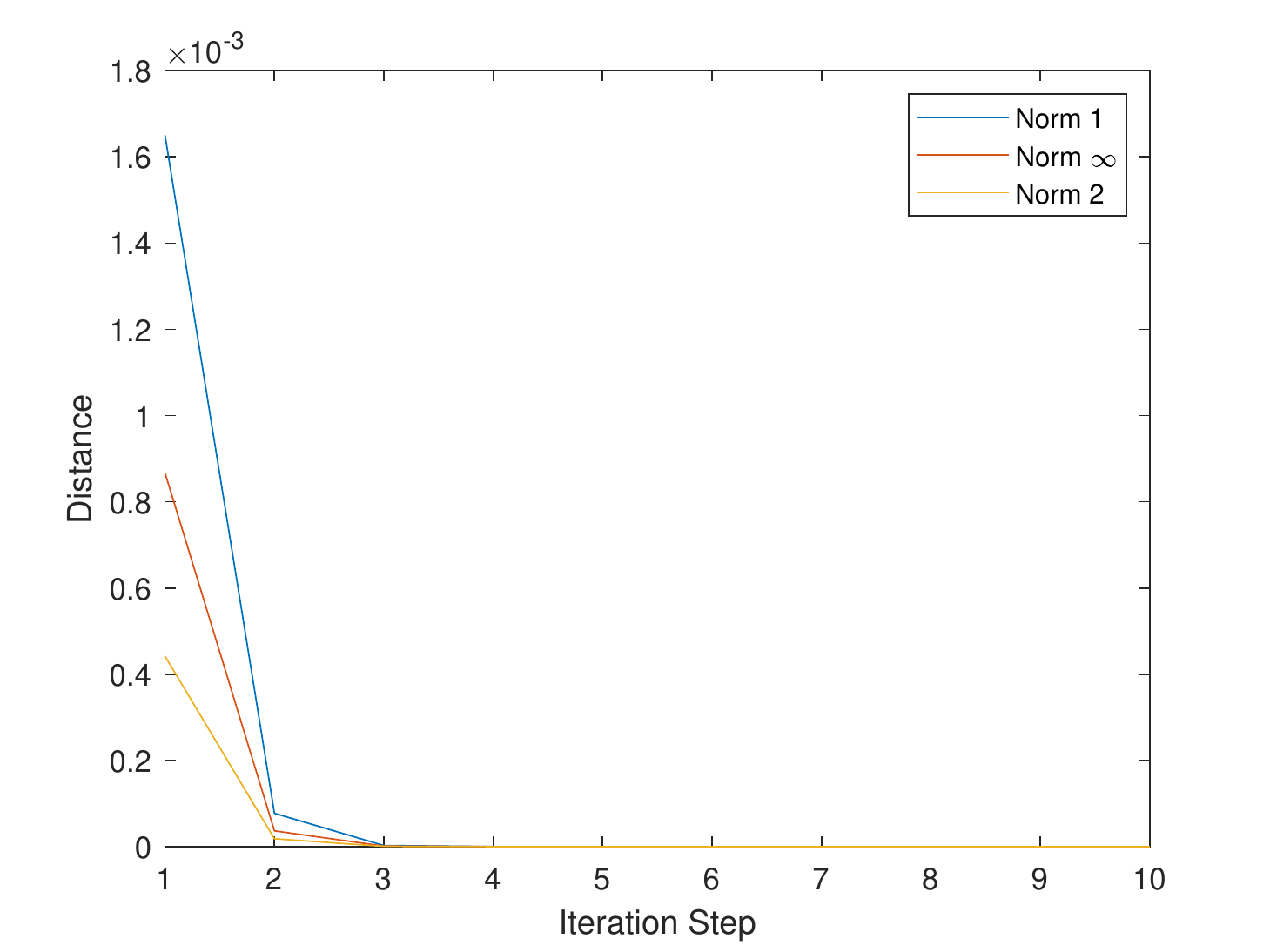}
\end{subfigure}\hspace*{\fill}
\begin{subfigure}[t]{0.48\textwidth}\centering
\footnotesize
\begin{tabular}{llll}
Iteration & Norm 1 & Norm 2 & Norm $\infty$ \\
\hline
1 & 0.0017 & 4.4E-4 & 8.6E-4\\
2 & 7.7E-5 & 1.8E-5 & 3.6E-5\\
3 & 2.6E-5 & 6.9E-7 & 1.2E-6\\
4 & 8.3E-8 & 2.1E-8 & 3.6E-8\\
5 & 2.4E-9 & 6E-10 & 9E-10\\
6 & 6E-11 & 1E-11 & 2E-11\\
7 & 1E-12 & 3E-13 & 5E-13\\
8 & 2E-14 & 8E-15 & 1E-14\\
9 & 7E-16 & 2E-16 & 2E-16\\
10 & 3E-16 & 1E-16 & 1E-16\\
\hline
\end{tabular}
\end{subfigure}
\caption{Convergence Analysis - Testing 2}\label{convergence_testing2}
\end{figure}
\begin{figure}[hbtp]
\begin{subfigure}{0.48\textwidth}\centering
\includegraphics[scale=0.55]{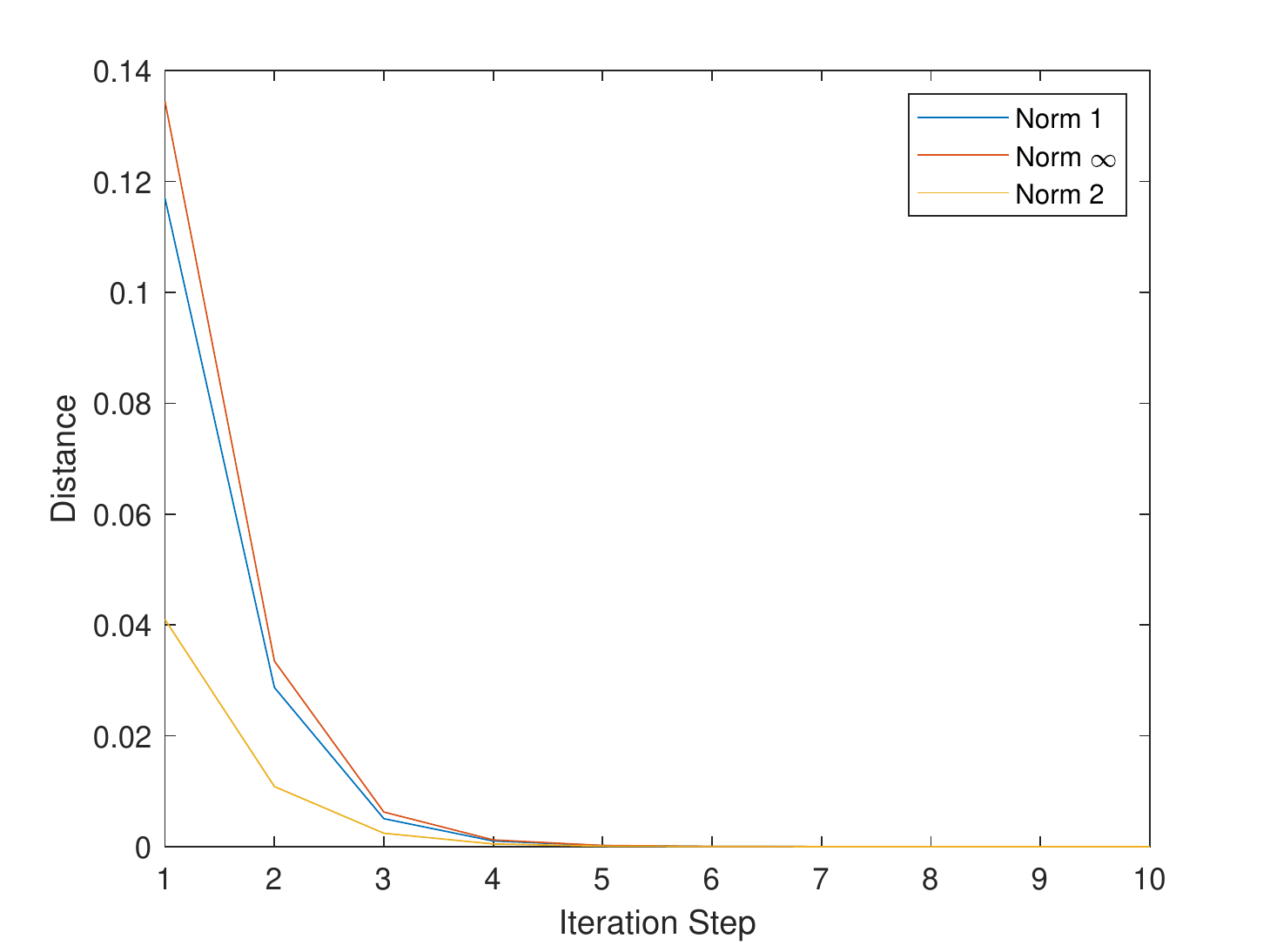}
\end{subfigure}\hspace*{\fill}
\begin{subfigure}[t]{0.48\textwidth}\centering
\footnotesize
\begin{tabular}{llll}
Iteration & Norm 1 & Norm 2 & Norm $\infty$ \\
\hline
1 & 0.1169 & 0.0410 & 0.1344\\
2 & 0.0287 & 0.0108 & 0.0335\\
3 & 0.0050 & 0.0024 & 0.0062\\
4 & 0.0010 & 4.8E-4 & 0.0012\\
5 & 1.7E-4 & 8.6E-5 & 2.0E-4\\
6 & 2.6E-5 & 1.4E-5 & 3.0E-5\\
7 & 3.7E-6 & 2.1E-6 & 4.2E-6\\
8 & 5.0E-7 & 3.0E-7 & 5.5E-7\\
9 & 6.7E-8 & 4.0E-8 & 6.9E-8\\
10 & 8.5E-9 & 5.0E-9 & 8.3E-9\\
\hline
\end{tabular}
\end{subfigure}
\caption{Convergence Analysis - Testing 3}\label{convergence_testing3}
\end{figure}


\section{Conclusion}

\noindent In this paper we studied the nonlinear partial differential equation that explains the dynamic of a financial option under a Black-Scholes' model with transaction costs. We extended the general literature on this subject by generalizing the dimension of the option (i.e multi-asset option) and allowing different transaction costs function. Following Perron methodology, we proved the existence of a viscosity solution by finding a proper set of sub and supersolutions of the original problem. Furthermore, we developed a numerical procedure to find an approximate strong solution by following an iterative method. For this purpose, an ADI scheme was developed in order to deal with mixed derivatives and to work under a finite difference approach. Nonetheless, we provided numerical examples by setting different possible asset prices, volatilities and interest rates among others to study how the ADI framework performs and how sensitive the output is to changes in the delta hedging time step. Different expected results are observed after running the simulations. Firstly, as transaction costs are proportional to the second derivatives of the option price, the transaction costs function reaches its maximum near the at-the-money region. Secondly, it is seen that the transaction costs function explodes when the frequency of rebalancing the replicant portfolio tends to infinity (and $\Delta_{TC}$ goes to zero). Finally, we observe that given the three proposed testing frameworks, the iterative method converges after less than seven iterations.\\

\section{Acknowledgement}

\noindent This work was partially supported by project CONICET PIP 11220130100006CO and project UBACYT 20020160100002BA.

\clearpage
\appendix
\section{Differential matrix calculation steps}\label{appendix}

\noindent Result \eqref{DF_1} follows from these steps:

\begin{align}
\frac{\partial}{\partial B_{kl}}  \, tr\left( \frac{1}{2}  A \, B \right) &=   \frac{1}{2} \sumai \sumaj A_{ij} \frac{\partial B_{ji}}{\partial B_{kl}}, \nonumber\\
&=  \frac{1}{2} \sumai \sumaj A_{ij} \delta_{jk} \delta_{il}, \nonumber \\
&=  \frac{1}{2} A_{lk} =  \frac{1}{2} A_{kl}.
\end{align}

\noindent Result \eqref{DF_2} follows from these steps:

\begin{align}
\frac{\partial}{\partial B_{lm}} \sqrt{\sumaj \sumak B_{ij} \, A_{jk} \, B_{ki}} &= \frac{1}{2}\left(\sumaj \sumak B_{ij} \, A_{jk} \, B_{ki}\right)^{-1/2} \frac{\partial}{\partial B} \sumaj \sumak B_{ij} \, A_{jk} \, B_{ki} \nonumber \\
&= \frac{1}{2}\left(\sumaj \sumak B_{ij} \, A_{jk} \, B_{ki}\right)^{-1/2} \sumaj \sumak \left( \frac{\partial B_{ij}}{\partial B_{lm}} A_{jk} B_{ki} + B_{ij} A_{jk} \frac{\partial B_{ki}}{\partial B_{lm}} \right) \nonumber \\
&= \frac{1}{2}\left(\sumaj \sumak B_{ij} \, A_{jk} \, B_{ki}\right)^{-1/2} \sumaj \sumak \left( \delta_{il} \delta_{jm} A_{jk} B_{ki} + B_{ij} A_{jk} \delta_{kl} \delta_{im} \right) \nonumber \\
&= \frac{1}{2}\left(\sumaj \sumak B_{ij} \, A_{jk} \, B_{ki}\right)^{-1/2} \left[ \sumak A_{mk} B_{kl} + \sumaj B_{mj} A_{jl}\right] \nonumber \\
&= \frac{1}{2}\left(\sumaj \sumak B_{ij} \, A_{jk} \, B_{ki}\right)^{-1/2} \sumak A_{mk} B_{kl} + \sumaj B_{mk} A_{kl} \nonumber \\
&=  \frac{1}{2}\left(\sumaj \sumak B_{ij} \, A_{jk} \, B_{ki}\right)^{-1/2} \left[ \left(AB\right)_{ml} + \left(BA\right)_{ml} \right].
\end{align}

\noindent If we denote 

\begin{align}
H_i\left(y\right) = \sqrt{2 \, \Delta t \, \sumaj \sumak B_{ij} \, A_{jk} \, B_{ki}}  \, y
\end{align}

\noindent then, result \eqref{DF_3} follows from these steps:

\begin{align}
\frac{\partial}{\partial B_{lm}} C\left(H_i\left(y\right) \right) &= C'\left(H_i\left(y\right) \right) \frac{\partial}{\partial B_{lm}} H_i\left(y\right)  \nonumber \\
&= C'\left(H_i\left(y\right) \right) \, y \, \frac{1}{2} \left(2 \Delta t \Theta_i \right)^{-1/2} \, 2 \, \Delta t \frac{\partial}{\partial B_{lm}} \sumaj \sumak B_{ij} A_{jk} B_{ki} \nonumber \\
&= C'\left(H_i\left(y\right) \right) \, y \, \frac{1}{2} \left(2 \Delta t \Theta_i \right)^{-1/2} \, 2 \, \Delta t \left[ AB + BA \right] 
\end{align}

\clearpage
\addcontentsline{toc}{chapter}{Bibliography}

\bibliography{viscosity_version.bbl}

\begin{thebibliography}{10}

\bibitem{amster2005black}
P~Amster, CG~Averbuj, MC~Mariani, and D~Rial.
\newblock A {Black--Scholes} option pricing model with transaction costs.
\newblock {\em Journal of Mathematical Analysis and Applications},
  303(2):688--695, 2005.

\bibitem{barles1998option}
Guy Barles and Halil~Mete Soner.
\newblock Option pricing with transaction costs and a nonlinear black-scholes
  equation.
\newblock {\em Finance and Stochastics}, 2(4):369--397, 1998.

\bibitem{black1973pricing}
Fischer Black and Myron Scholes.
\newblock The pricing of options and corporate liabilities.
\newblock {\em The Journal of Political Economy}, pages 637--654, 1973.

\bibitem{boyle1992option}
Phelim~P Boyle and Ton Vorst.
\newblock Option replication in discrete time with transaction costs.
\newblock {\em The Journal of Finance}, 47(1):271--293, 1992.

\bibitem{davis1993european}
Mark~HA Davis, Vassilios~G Panas, and Thaleia Zariphopoulou.
\newblock European option pricing with transaction costs.
\newblock {\em SIAM Journal on Control and Optimization}, 31(2):470--493, 1993.

\bibitem{grandits2001leland}
Peter Grandits and Werner Schachinger.
\newblock Leland's approach to option pricing: The evolution of a
  discontinuity.
\newblock {\em Mathematical Finance}, 11(3):347--355, 2001.

\bibitem{grossinho2009note}
MR~Grossinho and E~Morais.
\newblock A note on a stationary problem for a {Black-Scholes} equation with
  transaction costs.
\newblock {\em Int. J. Pure Appl. Math}, 51:579--587, 2009.

\bibitem{haug2007complete}
Espen~Gaarder Haug.
\newblock {\em The complete guide to option pricing formulas}.
\newblock McGraw-Hill Companies, 2007.

\bibitem{hoggard1994hedging}
T~Hoggard, AE~Whalley, and P~Wilmott.
\newblock Hedging option portfolios in the presence of transaction costs.
\newblock {\em Advances in Futures and Options Research}, 7(1):21--35, 1994.

\bibitem{imai2006hoggard}
Hitoshi Imai, Naoyuki Ishimura, Ikumi Mottate, and Masaaki Nakamura.
\newblock On the {Hoggard}-{Whalley}-{Wilmott} equation for the pricing of
  options with transaction costs.
\newblock {\em Asia-Pacific Financial Markets}, 13(4):315--326, 2006.

\bibitem{imbert2013introduction}
Cyril Imbert and Luis Silvestre.
\newblock An introduction to fully nonlinear parabolic equations.
\newblock In {\em An introduction to the K{\"a}hler-Ricci flow}, pages 7--88.
  Springer, 2013.

\bibitem{in2010adi}
KJ~In't~Hout and S~Foulon.
\newblock Adi finite difference schemes for option pricing in the {Heston}
  model with correlation.
\newblock {\em International journal of numerical analysis and modeling},
  7(2):303--320, 2010.

\bibitem{in2007stability}
KJ~In't~Hout and BD~Welfert.
\newblock Stability of {ADI} schemes applied to convection--diffusion equations
  with mixed derivative terms.
\newblock {\em Applied numerical mathematics}, 57(1):19--35, 2007.

\bibitem{jeong2013comparison}
Darae Jeong and Junseok Kim.
\newblock A comparison study of {ADI} and operator splitting methods on option
  pricing models.
\newblock {\em Journal of Computational and Applied Mathematics}, 247:162--171,
  2013.

\bibitem{leland1985option}
Hayne~E Leland.
\newblock Option pricing and replication with transactions costs.
\newblock {\em The Journal of Finance}, 40(5):1283--1301, 1985.

\bibitem{lepinette2012modified}
Emmanuel Lepinette.
\newblock Modified {Leland’s} {Strategy} for a {Constant} {Transaction}
  {Costs} {Rate}.
\newblock {\em Mathematical Finance}, 22(4):741--752, 2012.

\bibitem{mckee1970alternating}
S~McKee and AR~Mitchell.
\newblock Alternating direction methods for parabolic equations in two space
  dimensions with a mixed derivative.
\newblock {\em The Computer Journal}, 13(1):81--86, 1970.

\bibitem{mckee1996alternating}
S~McKee, DP~Wall, and SK~Wilson.
\newblock An alternating direction implicit scheme for parabolic equations with
  mixed derivative and convective terms.
\newblock {\em Journal of Computational Physics}, 126(1):64--76, 1996.

\bibitem{vsevvcovivc2016analysis}
Daniel {\v{S}}ev{\v{c}}ovi{\v{c}} and Magdal{\'e}na {\v{Z}}it{\v{n}}ansk{\'a}.
\newblock Analysis of the nonlinear option pricing model under variable
  transaction costs.
\newblock {\em Asia-Pacific Financial Markets}, pages 1--22, 2016.

\bibitem{zakamouline2008hedging}
Valeri Zakamouline.
\newblock Hedging of option portfolios and options on several assets with
  transaction costs and nonlinear partial differential equations.
\newblock {\em International Journal of Contemporary Mathematical Sciences},
  3(4):159--180, 2008.

\bibitem{zakamulin2008option}
Valeriy Zakamulin.
\newblock Option pricing and hedging in the presence of transaction costs and
  nonlinear partial differential equations.
\newblock {\em Available at SSRN 938933}, 2008.

\end{thebibliography}
\bibliographystyle{plain}

\end{document}